\documentclass[11pt,a4paper]{article}

\usepackage[margin=1in]{geometry}
\usepackage{amsmath,amssymb,amsthm}
\usepackage{mathpartir}          
\usepackage{stmaryrd}            
\usepackage{xspace}              
\usepackage{listings}            
\usepackage{booktabs}            
\usepackage{subcaption}          
\usepackage{hyperref}            
\usepackage{cleveref}            
\usepackage[numbers]{natbib}              
\usepackage{xcolor}              
\usepackage{graphicx}
\usepackage{aliascnt}
\usepackage{placeins}
\usepackage{tabularx}
\usepackage{booktabs}
\usepackage{doi}

\hypersetup{
  colorlinks=true,
  linkcolor=blue!70!black,
  citecolor=green!50!black,
  urlcolor=blue!70!black,
}

\usepackage{aliascnt}

\crefname{theorem}{theorem}{theorems}
\Crefname{theorem}{Theorem}{Theorems}

\newaliascnt{lemma}{theorem}
\newtheorem{lemma}[lemma]{Lemma}
\aliascntresetthe{lemma}
\crefname{lemma}{lemma}{lemmas}
\Crefname{lemma}{Lemma}{Lemmas}

\newaliascnt{corollary}{theorem}

\aliascntresetthe{corollary}
\crefname{corollary}{corollary}{corollaries}
\Crefname{corollary}{Corollary}{Corollaries}

\newaliascnt{proposition}{theorem}

\aliascntresetthe{proposition}
\crefname{proposition}{proposition}{propositions}
\Crefname{proposition}{Proposition}{Propositions}

\theoremstyle{definition}
\newaliascnt{definition}{theorem}
\newtheorem{definition}[definition]{Definition}
\aliascntresetthe{definition}
\crefname{definition}{definition}{definitions}
\Crefname{definition}{Definition}{Definitions}

\newaliascnt{invariant}{theorem}

\aliascntresetthe{invariant}
\crefname{invariant}{invariant}{invariants}
\Crefname{invariant}{Invariant}{Invariants}

\theoremstyle{remark}
\newaliascnt{remark}{theorem}
\newtheorem{remark}[remark]{Remark}
\aliascntresetthe{remark}
\crefname{remark}{remark}{remarks}
\Crefname{remark}{Remark}{Remarks}

\newcommand{\AR}{\textsc{Agentic Redux}\xspace}
\newcommand{\REPO}{\url{https://github.com/Thistleseeds/agentic-redux}\xspace}

\newcommand{\project}{\mathsf{project}}

\newcommand{\State}{\mathsf{State}}

\newcommand{\Slice}{\mathsf{Slice}}
\newcommand{\Rejection}{\mathsf{Rejection}}

\newcommand{\Domain}{\ensuremath{{\mathsf{DOMAIN}}}}

\newcommand{\stateAt}[1]{\State@#1}

\newtheorem{procedure}{Procedure}
\crefname{procedure}{procedure}{procedures}
\Crefname{procedure}{Procedure}{Procedures}

\newcommand{\Occ}{\mathsf{Occ}}
\newcommand{\Time}{\mathsf{Time}}

\newcommand{\Action}{\mathsf{Action}}
\newcommand{\Roles}{\mathsf{Roles}}
\newcommand{\Fields}{\mathsf{Fields}}
\newcommand{\Footprint}{\mathsf{footprint}}
\newcommand{\Approval}{\mathsf{Approval}}
\newcommand{\Escalation}{\mathsf{Escalation}}
\newcommand{\Outcome}{\mathsf{Outcome}}
\newcommand{\InvResult}{\mathsf{InvResult}}
\newcommand{\AuditEntry}{\mathsf{AuditEntry}}
\newcommand{\ReviewQueue}{\mathsf{ReviewQueue}}
\newcommand{\PendingProposal}{\mathsf{PendingProposal}}
\newcommand{\FrameworkState}{\mathsf{FState}}
\newcommand{\AdjudicablePlus}{\ensuremath{\Adjudicable^{+}}}
\newcommand{\AgReduxPlus}{\ensuremath{\AgRedux^{+}}}
\newcommand{\Counselor}{\mathsf{Counselor}}
\newcommand{\CounselorOutcome}{\mathsf{CounselorOutcome}}
\newcommand{\cCommitted}{\mathsf{cCommitted}}
\newcommand{\cRejected}{\mathsf{cRejected}}
\newcommand{\awaitCounselor}{\mathsf{awaitCounselor}}
\newcommand{\authorized}{\mathrm{authorized}}

\newcommand{\Adjudicable}{\textsc{Adjudicable}}
\newcommand{\AgRedux}{\mathsf{AgenticRedux}}

\newcommand{\Agent}{\mathsf{Agent}}

\newcommand{\AgentPop}{\Pi}
\newcommand{\Coord}{\pi}
\newcommand{\Commit}{\kappa}
\newcommand{\Oracle}{\Omega}
\newcommand{\Arch}{\mathcal{A}}

\newcommand{\invariants}{\mathsf{invariants}}
\newcommand{\applyMutation}{\mathsf{applyMutation}}

\newcommand{\pass}{\mathsf{pass}}
\newcommand{\rejectRes}{\mathsf{reject}}
\newcommand{\escalateRes}{\mathsf{escalate}}
\newcommand{\approved}{\mathsf{approved}}
\newcommand{\rejected}{\mathsf{rejected}}
\newcommand{\escalated}{\mathsf{escalated}}
\newcommand{\config}[3]{\langle #1,\, #2,\, #3 \rangle}

\newcommand{\restr}{\!\upharpoonright\!}   

\title{Provably Auditable and Safe LLM Agents from Human-Authored Ontologies
\thanks{A supporting repository with code for the Agentic Redux 
kernel and all Problem Domain examples can be found at \REPO}}
\author{Aaron Sterling\\Thistleseeds\\\texttt{asterling@thistleseeds.com}}
\date{}

\begin{document}

\maketitle

\begin{abstract}
We introduce the LLM agent architecture Agentic Redux, intended for use with nontrivial problem domains that require linear auditability.
Using the typed lambda calculus, we prove that, run on appropriate domains, Agentic Redux executions are semantically guaranteed to be correct,
with all decisions recorded in an append-only ledger.
We present two production-grade appropriate domains, in healthcare billing compliance, and security vulnerability disclosure.
Working code for Agentic Redux run on both domains is available in a supporting code repository.
We also introduce Ontology-First Agent Design, a methodology for creation of agent frameworks on a problem domain, 
in which a human expert ontologizes the problem domain with Basic Formal Ontology, 
and then assigns an LLM to derive roles that agents and humans-in-the-loop can fill, in order to work the problems in the domain.
\end{abstract}

\section{Introduction}
\label{sec:intro}
\subsection{Semantic Verification of Safety}
\label{sec:intro-safety}

LLMs hallucinate.
Even when they respond correctly, they behave nondeterministically.
In safety-priority domains that require linear auditability, like finance or healthcare compliance, agent behavior can be hard to manage.
Currently, engineers manage the safety of agent behavior operationally.
For example, Anthropic's Managed Agents \cite{anthropicmanagedagents} make it easy to swap out a stale agent harness for an up-to-date harness.
The primary goal of this paper is to show that, for some problem domains, using the correct agent architecture means you never have to swap out the harness, ever.
The \emph{type-theoretic properties} of the architecture provide a \emph{semantic guarantee of safety}.

Programming language theory, and the typed lambda calculus, provide a sophisticated toolbox to prove that bad behaviors can never happen.
That said, while the theorems in this paper prove that our executions of interest are always correct, real world software is buggy.
I recommend that builders follow a defense-in-depth approach of wrapping each agent in a governance layer (such as the layer provided by Microsoft Agent Framework\cite{microsoftagentframework}) that watches for exactly the behavior the type system prevents.

While Managed Agents provides a runtime substrate, and Microsoft Agent Framework provides guardrails, this paper focuses on a third layer: agent \emph{architecture}, the rules about what types of agents can exist, what each agent can see and do, how agents communicate with one another, and how state changes happen.

To see why agent architecture matters, let's consider a naive approach. Define the \emph{Microservices Architecture} as the architecture where each agent behaves independently of all other agents, and all agents write to a shared log.
A two-agent system for which the Microservices Architecture fails is the following.
The system has a budget of \$100,000.
Agent A can either do nothing or spend \$45,000.
Agent B can either do nothing or spend \$60,000.
Each agent sees only its local state.
Since, in the Microservices Architecture, both agents act independently of each other, 
it is within the rules for both agents to spend money simultaneously, which causes the system to exceed its \$100,000 budget.
(Berenson \emph{et al}.\cite[]{berenson} termed this kind of failing execution a \emph{Write Skew}.)
If, instead, the two agent system were running an architecture in which Agents A and B 
could propose actions to a meta-agent, and the meta-agent with view of the global state would adjudicate the proposals, 
then the meta-agent could ensure that the system would never exceed the budget.

Robust microservices architectures mitigate Write Skew with error-reduction patterns (\emph{e.g.}, sagas).
Similarly, current agent frameworks are managing problems like Write Skew either by trivializing inter-agent coordination 
to sidestep the problem; or by gatekeeping state changes in ways that they hope will work, such as asking an LLM 
whether the state change is OK, then following the LLM's recomendation, with no consequence if the recommendation is incorrect.
The second approach is the ``LLM as Judge'' pattern of \cite{lakshmanan2025generative}.
Those approaches are used by all the multi-agent apps in \verb|awesome-llm-apps|\cite{awesomellmapps}, 
a public Github repository with over 100,000 stars.
A stronger approach would be to employ architecture with a mathematical guarantee that
 any change to the global state preserves ``nice properties.''

The rule (nice property), ``The system can spend at most \$100,000,'' is an \emph{invariant} of the system: 
it is a statement that must be always true about the global state.
To prevent Write Skew (and other bad behaviors), the agent architecture must preserve invariants when the global state transitions.
That \emph{preservation of invariants} can be shown to be a semantic guarantee of a correctly-typed architecture, 
ensuring the agent framework complies with, \emph{e.g.}, regulatory requirements that must always be true.

Real-world agent frameworks almost always deploy with a human-in-the-loop mechanism for a person in
charge to act as supervisor or problem-solver. All theorems in this paper about pure Agentic Redux
can be modified to cover Agentic Redux with Counselor Queue, where the Counselor is a human with authority over the system.
In particular, both preservation of invariants and linear auditability still hold.
For more on this topic, see \cref{sec:proofs-counselor-queue}.
\subsection{Threat Model (What is Safety)}
\label{sec:threat-model}
For purposes of this paper, I scope the notion of safety as follows.

\emph{In scope:} domain modules may contain bugs or misspecifications, 
agents may produce arbitrary proposals (whether from LLM hallucination, tool error, prompt injection,
or deliberate adversarial construction); 
and, in the proof machinery, the oracle $\omega$ may supply arbitrary scheduling and response sequences. 
Theorem 1 guarantees that, regardless of these failures, invariants are always preserved. 
Theorem 2 guarantees that every decision is recorded faithfully, 
so the system can be linearly audited at any time. 

\emph{Out of scope:} if an invariant is poorly specified, and does not support a real-world policy, the system will enforce it,
but misbehave with respect to real-world application.
While humans in the loop can make any policy decision, including commits that violate invariants,
their decisions must be formatted in a way that obeys the system's type contracts.
Confidentiality of the audit log itself is not in scope.
Some decisions in Security Vulnerability Disclosure depend on external signals,
like assignment of an ID number to a CVE;
recording those signals faithfully is the responsibility of the agent harness,
which is not in scope.
This paper focuses on safety requirements (bad things don't happen);
I consider liveness requirements (good things eventually happen) in \cite{semantic-correctness-2}.
%
\subsection{Agentic Redux}
\label{sec:intro-redux}
\subsubsection{Agentic Redux Architecture}
\label{sec:agentic-redux-architecture}
This paper introduces the agent architecture \emph{Agentic Redux}.
It is inspired by the front-end application state manager Redux\cite{redux}, created by Dan Abramov.
Abramov wanted to ensure that the most recent, fresh state was always the state displayed to the user.
Imagine you have the same web page open in two browser tabs.
You click a check box in one tab, then switch to the other tab.
Does the second tab display the checked box, or does it still display the stale state with the box unchecked?
The Redux solution to this problem is to make all components presentation-only, and to centralize
 all state changes through a single decisonmaker that can see the global state of the app.
A popular Redux saying is, \emph{``Smart container, dumb components.''}

In Agentic Redux, each agent sees only their local slice of the global state, and computes their own evaluation function based on that local state. 
The evaluation could be simple and synchronous, or complex and asynchronous, invoking LLMs or other tools.
If an agent wants to change its local slice of state (the analogue to displaying a checked box because the user clicked an unchecked box), 
it proposes that change to a meta-agent.
The meta-agent adjudicates all proposals from local agents.
The meta-agent's adjuducation function is written to preserve desired invariants of the global state.
The meta-agent can see the global state, and is the only entity in the system with the power to change the global state.
If the meta-agent changes the global state, it then sends each agent their local slice of the new global state.
Agents propose; the meta-agent decides and directs. See \cref{fig:agentic-redux-architecture} for a diagram.

\begin{figure}[t]
     \centering
     \includegraphics[width=\linewidth]{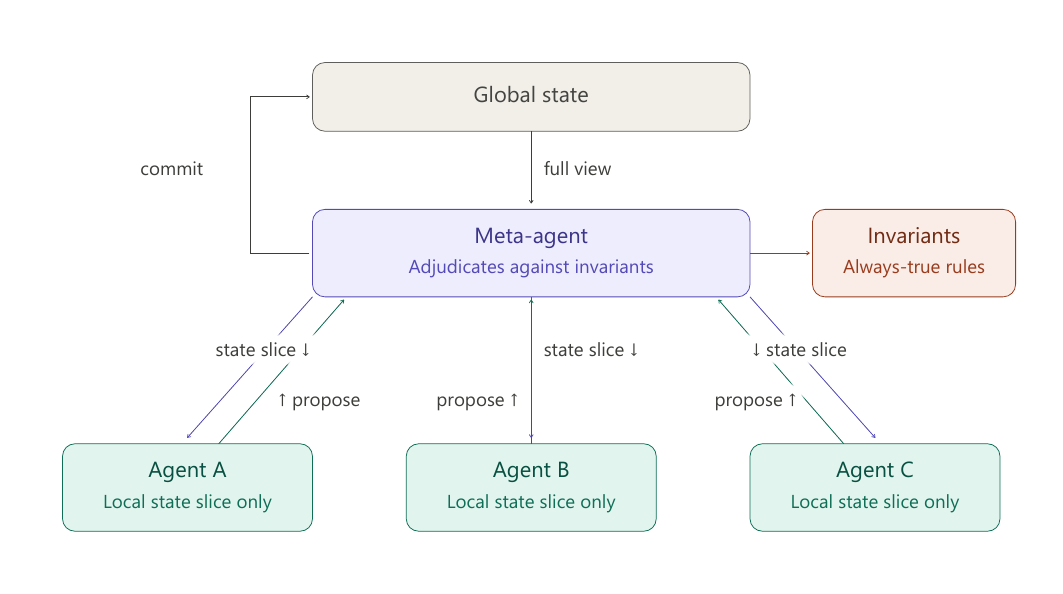}
     \caption{The \AR architecture. Sub-agents receive isolated state slices and propose actions to a meta-agent, 
              which adjudicates proposals against declared invariants before committing state transitions. Code is available at \REPO}
     \label{fig:agentic-redux-architecture}
   \end{figure}
\subsubsection{Linear Auditability}
\label{sec:linear-auditability}
The meta-agent keeps a log.
Each time the meta-agent declines a proposal from one of the subagents, 
the meta-agent records the denial, and why it denied the proposal (which invariant would not have been preserved if it had accepted the proposal).
Each time the meta-agent accepts a proposal and changes the global state, it logs the new global state and proof that all invariants are preserved.
This log provides an audit trail, ordered linearly in time, so every decision of the system can be reviewed by an external auditor.

Linear auditability benefits from other properties Agentic Redux provides:
the log is a \emph{ledger}, meaning no data is ever removed, and writing is append-only;
only the meta-agent writes to the log, and only when deciding whether to change the global state;
every decision by the meta-agent about whether to change the global state produces exactly one log entry.
The formal proof of all these benefits is discussed in \cref{sec:proofs-about-agentic-redux}.
Later in the paper, we will see architectures that allow either the meta-agent or a human-in-the-loop to write to the log.
\subsubsection{Agentic Redux May Be a ``Natural'' Architecture}
\label{sec:agentic-redux-natural-architecture}
I created Agentic Redux to handle issues I was running into when designing software for healthcare billing compliance. 
To provide evidence that Agentic Redux may be a ``natural'' agentic architecture, 
with application to a variety of problem domains, Agentic Redux appears near-isomorphic 
to the Risk Agents framework\cite{riskagents}, which was independently created by engineers at CashApp.
While CashApp created Risk Agents to solve problems in financial compliance, 
I stumbled across Risk Agents when reading minutes of a working group meeting that 
was discussing how to apply Risk Agents to automated social media moderation.

For reasons of space, I will defer discussion of social media moderation until future work.
However, the supporting code repository for this paper contains the full ontology and code examples 
for social media User Case Management, a problem domain that benefits from Agentic Redux (and from Risk Agents).
The ontology is derived in part from Osprey\cite{osprey}, an open-source social media moderation tool connected to that working group.

\subsection{Methodology}
\label{sec:intro-methodology}
\subsubsection{Domains and Functors}
\label{sec:intro-domains-functors}
To study the architecture separately from the problem it is helping to solve, 
I use a Programming Language Theory technique of separating a working system into three pieces: 
a \emph{Domain Module}, an \emph{Architecture Functor}, and a \emph{Client Program}.

The composition of the Domain Module and the Architecture Functor can be thought of as a module making a call to an operating system kernel.
In the case of Agentic Redux, a domain provides invariants, the initial state, 
the meta-agent's adjudication function, and the evaluate-and-propose functions run by the local agents.
The meta-agent's adjudication function is a synchronous checklist of invariants to preserve
in all global state changes,
with the subagent functions can be either synchronous or asynchronous.
A local agent, using its own slice of state, may consult an LLM (or any other tool) before deciding whether to propose an action.
The Agentic Redux kernel then runs the system, with the guardrail that, whenever global state might change, 
that change is allowed only if all invariants are preserved.

The Client Program handles everything else: sandbox, session liveness, harness, test scenarios, governance.
These features are important, but they are out of scope for this paper.
I will assume that all Client Program features are available and work.

I define the properties of the Architecture Functor with the \emph{typed lambda calculus}\cite{pierce2002types}, 
a formalization of programming language pseudocode.
The typed lambda calculus is well-studied, and its literature contains deep theorems.
\Cref{sec:overview-formal-proofs} shows how to prove that Agentic Redux is semantically guaranteed to preserve system invariants.

\subsubsection{Computational Ontology}
\label{sec:intro-ontology}
%
%
While a Domain Module is code that captures information about an area of interest, I will use \emph{Problem Domain} 
to refer to a real-world problem space that humans are trying to reason about.
I consider two Problem Domains in this paper: health care compliance, and security vulnerability disclosure.
These domains sound different from each another (and they are!), but they share common underlying structure.
Significantly, they both admit Write Skew executions if run on the Microservices Architecture, 
and they both avoid Write Skew executions if run on Agentic Redux.

I move from Problem Domain to Domain Module by performing an intermediate step: construction of an ontology.
While ontology dates back to the Greeks, the study of \emph{computational ontology} began in the 1970s 
as part of early AI research\cite{ontology-history}.
There are now many computational ontologies available, 
to provide machine-understandable structure to music, anatomy, and other disciplines.

My ontology of choice is the \emph{Basic Formal Ontology} (BFO)\cite{bfo}, 
because it was created to ontologize complex processes, including organic processes, like disease progression in medicine.
BFO is an ``upper ontology,'' meaning that it provides 
an abstract template for categories and relations that have been shown to be robust over time.
A subject matter expert starts with BFO and constructs a ``lower ontology,'' 
which is a structural description of a concrete Problem Domain that follows the rules of BFO.

If you are a subject matter expert interested in ontologizing your domain of expertise, 
I recommend the excellent book \cite{arp2015bfo}, which is now open access.

\subsubsection{Ontology-First Agent Design}
\label{sec:intro-onotlogy-fingerprint-architecture}
The methodology I adopted (which may be of interest independent of Agentic Redux) I have termed \emph{Ontology-First Agent Design}.
The methodology is: (1) ontologize a Problem Domain using BFO; (2) ask an LLM to assess the ontology for \emph{domain fingerprints} 
of interest; (3) if the domain has a fingerprint that Agentic Redux can help with, 
ask the LLM to convert the ontology into a Domain Module, i.e., into code; (4) run the Domain Module on Agentic Redux 
and observe the results, clarifying the ontology as needed.
See \cref{fig:ontology-first-agent-design} for a diagram.

\begin{figure}[t]
     \centering
     \includegraphics[width=\linewidth]{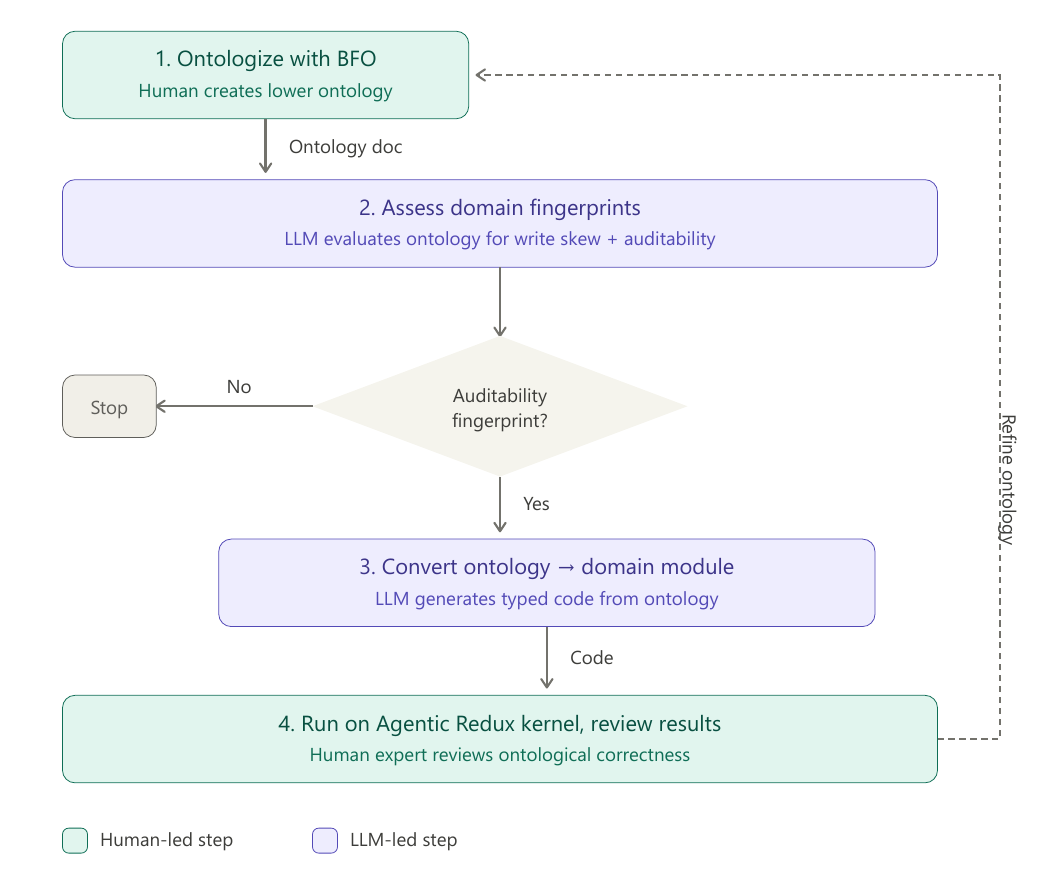}
     \caption{Ontology-First Agent Design. The ontologies referenced in this paper were produced by this method. Ontology files are available at \REPO}
     \label{fig:ontology-first-agent-design}
   \end{figure}

A domain fingerprint is an attribute of the Problem Domain that informs what is required to solve problems in that domain.
The fingerprint, ``All execution is seqential in a pipeline,'' makes Write Skew impossible, 
and is evidence that a simple architecture like the Microservices Architecture might be enough solve problems in that domain.
A fingerprint implying Write Skew indicates that a more sophisticated architecture is needed.
All Problem Domains considered in this paper imply Write Skew, and also have an \emph{Auditability Fingerprint}: 
the system must be able to report on every action the system performed, ordered linearly in time.
A Problem Domain with both Write Skew and Auditability fingerprints may be a strong candidate to benefit from Agentic Redux.
(A Problem Domain with just the Auditability fingerprint might be best run on a simple pipeline, if there are no nontrivial agent interactions.)

While I chose the two Problem Domains considered in this paper because they are structurally similar,
I ontologized each individually, and their ontologies are not a reskin.
Including the User Case Management ontology in the discussion for a moment, 
both the security vulnerability and social media moderation domains contained a timed governance relationship 
that the health care compliance domain did not; and a relationship present in both health care compliance 
and security vulnerabiity disclosure was absent in social media moderation.

Please consider this part of the paper a presentation of a case study of Ontology-First Agent Design.
Between this paper and its upcoming parts, I have successfully used ontologies to auto-derive agents for systems
with seven different problem domains and three different agent architectures (not just Agentic Redux).
It's a small sample size, but robust enough that I think it's likely someone else could benefit from it too.

\subsubsection{Human vs. LLM Ontology Creation}
\label{sec:human-llm-ontology-creation}
It is worth pausing for a moment to review ontology creation by humans and by LLMs, 
because there is empirical data that might look contradictory at first glance, but, in fact, paints a unifying picture.
LLMs are not as good as humans at ontology creation (sometimes called ``ontology learning''), 
as shown in \cite{LLMs40L_2025,llms4life,ontology_learning_llms}.
However, at least according to the OntoURL benchmarks, LLMs are better than humans at reasoning over an ontology that already exists\cite{OntoURL}.

Despite the previous results, the quality of LLM-generated ontologies 
can be higher than the quality of ontologies created by novice human engineers\cite{ontology_generation_llms}.
This is not a contradiction, because an LLM's ability to 1-shot ontology 
creation is directly related to how completely \emph{humans already ontologized the space through documentation}.
The need for humans to pre-ontologize the space can be seen in \cite{llm_adapt_domain}, 
which presented LLMs with well-structured gibberish, and the LLMs were unable to ontologize 
the gibberish, showing an inability to reason over semantic relations between concepts.

One goal of Ontology-First Agent Design is to focus human expert input where it is most needed: 
creation of the ontology (at the start), and refinements to the ontology to improve the system's functionality (a feedback loop at the end).
The LLM does the work in the middle, where it is most effective.

\section{From Ontology to Agents}
\label{sec:ontology-to-agents}

\subsection{Basic Formal Ontology}
\label{sec:bfo}

I will present just enough BFO to show how to define agents and system invariants, so Agentic Redux can run on those agents to preserve the invariants.
The most important BFO definitions for our purposes are independent continuant, process and role.

An \emph{independent continuant} is a thing with its own relations and qualities.
In the two ontologies considered in this paper, both have one primary independent continuant: 
a patient, or a security vulnerability.
The Security Vulnerability Ontology also has a secondary independent continuant,
which only plays a peripheral role as an occasional source of signals.

A \emph{process} is something that happens to an independent continuant, unfolding over time.
In the Security Vulnerability Disclosure ontology, there are twelve processes, including: 
DevelopPatch, ReviewPatch, ReleaseFix, ExecuteDisclosure.
Note the centrality of the independent continuant (the vulnerability) to each process:
the vulnerability has a patch developed for it, the fix for the vulnerability is released, until finally the vulnerability is disclosed.

A \emph{role} is an externally grounded, optional, property of an independent continuant.
``John is an employee at ACME Corporation,'' is an example of a role.
John, an independent continuant, holds the role relative to an independent continuant external to himself, and John would continue being John if he were no longer employed there.
We can extend the original ontology by introducing \emph{workers}: new independent continuants that bear ``employment'' roles grounded in the processes of the original ontology.
In the examples we consider in this paper, these derived ``employment roles'' or ``worker roles'' are separate from, and in addition to, any roles that may have existed in the original ontology.

A technical note: I strongly recommend you upload the pdf of BFO 2.0\cite{BFO2} to your LLM's knowledge base, instead of the BFO ISO standard\cite{BFO_ISO_standard} in OWL or OWL DL.
The OWL DL spec is the more useful format for tools like automated reasoners, but it lacks the lexical context that LLMs need when helping to create ontologies, as already discussed in \cref{sec:human-llm-ontology-creation}.
\subsection{Derivation of Agent Roles}
\label{sec:derivation-roles}
Given an ontology of a problem domain, the roles we add to that ontology are the employment opportunities available to work on problems in the domain.
Those roles are then filled either by agents, or by humans in the loop. 
Each role comes with a job description (which processes to perform) and a security clearance (which slice of the global state is available).

I will now present a procedure (displayed visually in \cref{fig:role-derivation-procedure}) for determining which roles to create in order to extend a problem domain's ontology.
This is most of the work required to perform Step 3 (Convert Ontology to Domain Module) of Ontology-First Agent Design.
As shown in the color coding of \cref{fig:ontology-first-agent-design}, this is a procedure LLMs can perform well, in my experience.

\begin{procedure}[Role Derivation]
\label{proc:role-derivation}
List all processes in the ontology.
For each process, list everything the process needs to read or write; call that list the process's \emph{state footprint}.
Group processes whose state footprints substantially overlap.
For each group of processes, define a role whose state footprint is the union of the group's footprints.
Verify that no role touches state that its processes do not need; if that check fails, redo the previous step to define more fine-grained roles.
\end{procedure}

The preference for fine-grained roles in the Role Derivation procedure is an architectural best practice, not a logical requirement.
Just one agent could perform all roles, and that would avoid any Write Skew executions, but it would come with all the drawbacks of monolithic architecture.
Fine-grained roles reduce the blast radius of a faulty LLM step, and provide agent confinement, a type-level analogue to the \emph{Principle of Least Privilege} (PoLP).

PoLP is an access control property that is enforced by a governance layer to ensure a worker has access only to the resources needed to perform its work.
\emph{Agent confinement}, a generalization to agents of the classic Confinement Problem\cite{confinement}, restricts both what an agent can observe and what an agent can directly modify: an agent sees only the minimal slice of state needed to perform its work, and cannot directly modify state at all, since all state changes are mediated by the meta-agent.

\begin{figure}[t]
     \centering
     \includegraphics[width=\linewidth]{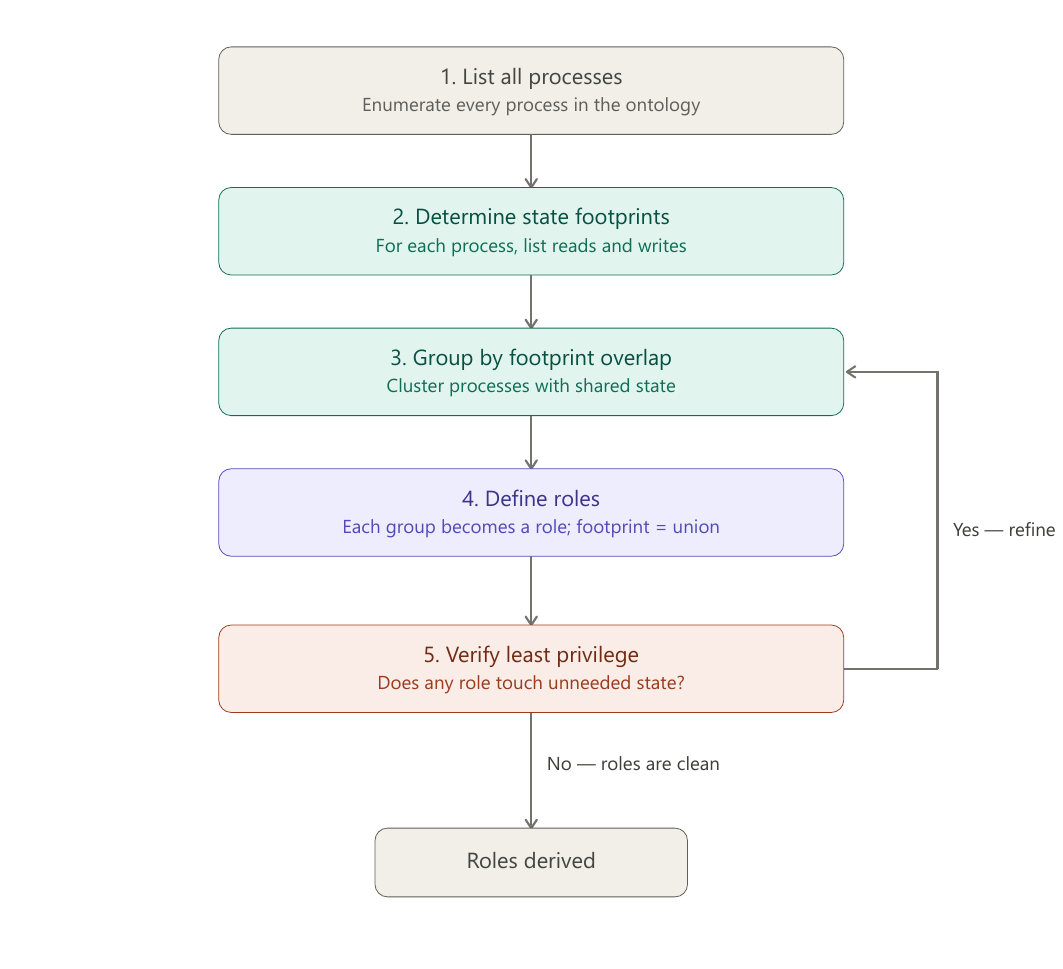}
     \caption{Role Derivation procedure.}
     \label{fig:role-derivation-procedure}
   \end{figure}

\subsection{Derivation of System Invariants}
\label{sec:system-invariants}
Once we have roles and slices of state that we can assign to workers, we need to determine the system invariants that we want the Agentic Redux meta-agent to enforce.
As with the Role Derivation procedure, I have found that LLMs are good at this step.
See \cref{fig:invariant-derivation-procedure} for a diagram.

Invariants come from the problem domain, not from the architecture.
The architecture's job is to \emph{enforce} invariants; the ontology's job is to \emph{declare} them.
BFO provides two natural sources of invariants: policy rules and process preconditions.

\begin{procedure}[Invariant Derivation]
\label{proc:invariant-derivation}
~

\noindent\textbf{Step 1: List all policy rules.}
Identify every \emph{information content entity} (ICE) in the ontology that represents a policy rule.
Each such ICE contributes at least one invariant.
List the state components each constraint references; call that list the invariant's \emph{state footprint}.

\noindent\textbf{Step 2: Extract process preconditions.}
For every process in the ontology, list the preconditions that must hold before the process can occur.
Each statement of type, ``Process P cannot occur unless these preconditions hold,'' is an invariant.
The invariant's state footprint includes the state written by any prerequisite process, and the state read by the gated process.
\end{procedure}
If you have an LLM derive the invariants, one human verifcation method would be to walk the life cycle of the ontology's independent continuant.
At every state transition, check that the conjunction of declared invariants is sufficient to prevent any transition the domain should prohibit.
If a transition is possible that the domain expert considers invalid but no invariant blocks, add the missing invariant as either a policy rule or a process prerequisite, then rerun the procedure.

Once you have listed all system invariants, classify them as either local or cross-cutting.
A \emph{local} invariant is an invariant whose state footprint is entirely contained within the state footprint of a worker.
A \emph{cross-cutting} invariant is an invariant whose state footprint intersects with the state footprints of at least two workers.
(A globally shared resource, like the \$100,000 budget in the system in the Introduction, is an example of a cross-cutting invariant.)
A local invariant can be enforced by the local worker whose state footprint contains it.
It is the meta-agent's responsibility, in Agentic Redux, to enforce that the cross-cutting invariants are preserved.

\begin{figure}[t]
  \centering
  \includegraphics[width=\linewidth]{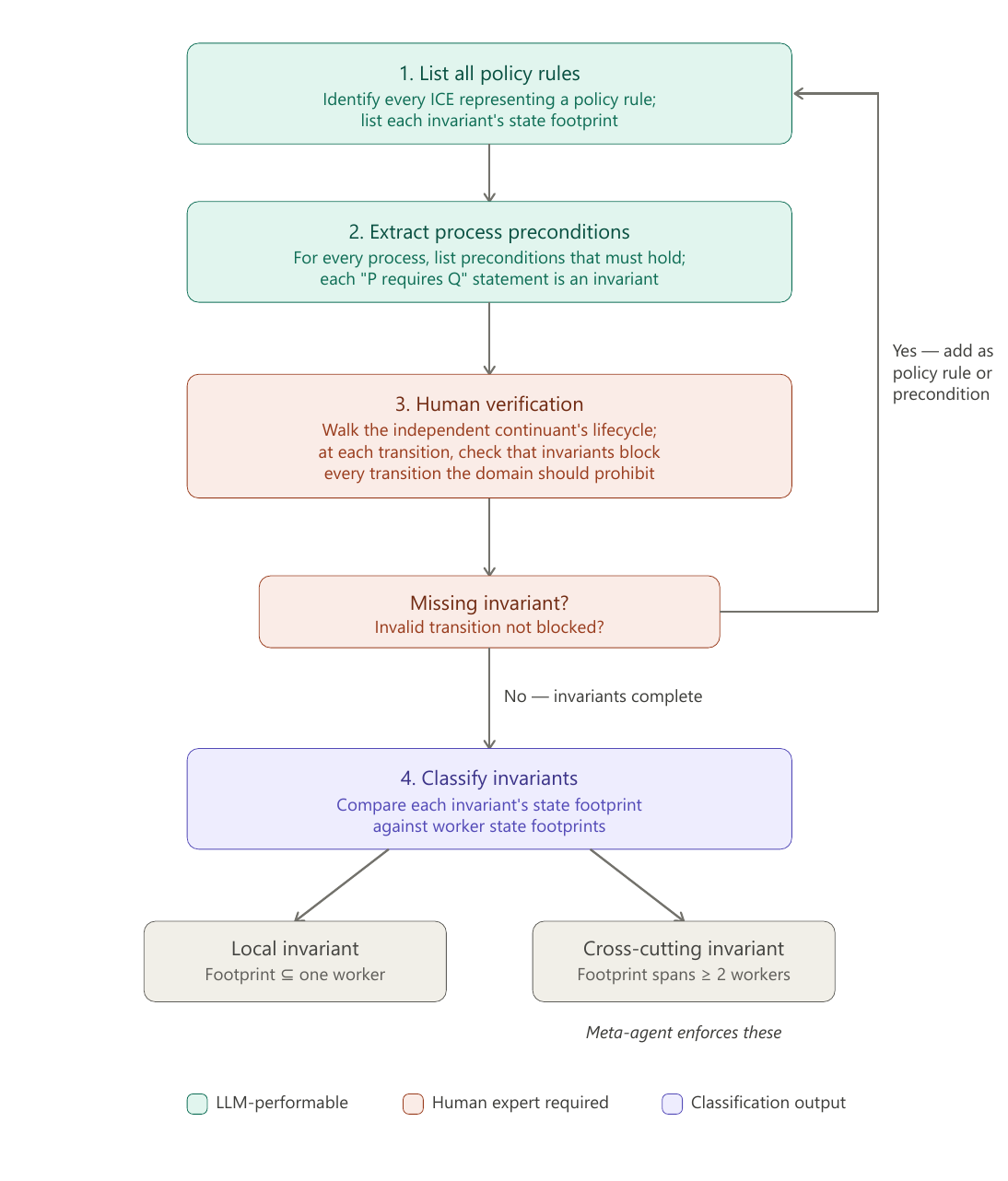}
  \caption{Deriving invariants, then assigning the cross-cutting invariants to the Agentic Redux meta-agent.}
  \label{fig:invariant-derivation-procedure}
\end{figure}
\subsection{Summary}
\label{sec:summary-ontology}
We've walked through how to derive worker roles and system invariants from a domain ontology.
(The ontologies available in the supporting code repository are BFO ontologies of problem domains, extended with roles and workers.)
In the ontologies considered in this paper, one worker is a human in the loop, while the rest of the workers are agents.
In other problem domains, perhaps it would make sense for all the workers to be agents.
The Agentic Redux meta-agent enforces preservation of a system invariant if that invariant affects the state footprint of at least two workers.

Now let's turn to the problem domains to see concrete examples.

\section{Motivating Problem Domains}
\label{sec:domains}
\subsection{UDT Compliance}
\label{sec:udt}

I created Agentic Redux to ensure compliance with billing requirements of Arizona's AHCCCS Program\cite{ahcccs-billing}.
These requirements can be difficult to parse, which can lead to service providers, or individuals with few resources, being denied reimbursement for services they thought were covered.
In particular, when a Urine Drug Test (UDT) is, or is not, covered, can be confusing.
My goal, when defining the UDT Compliance problem domain, was to be able to inform people with 100\% accuracy whether their UDT would be covered by AHCCCS.
This required coordination between billing and lab ordering, which led me to the construction of a meta-agent that could see both billing and lab ordering states.

There are two types of UDTs: a presumptive test, and a definitive test.
A \emph{presumptive test} uses immunoassay methods, and provides general results like, ``Opioids detected.''
A \emph{definitive test} uses a more specialized, more expensive method, to detect specific drugs, and can provide results like, ``Fentanyl at xyz concentration detected.''

The AHCCCS regulations prohibit an immunoassay test from being used to confirm a presumptive test; a definitive test is required to confirm the results of a presumptive test.

At the same time, the number of UDTs AHCCCS will cover in a week is dependent on an individual's \emph{abstinence tier}, which is based on the number of days since the individual's last positive test.
Someone with an abstinence tier of 90+ days, is eligible for three definitive tests in a 90-day period.
Someone with an abstinence tier of 0-30 days, is eligible for one definitive test every 7 days.

Putting everything together, we can see a billing pitfall that is structurally identical to a Write Skew execution.
A patient, call him David, is in the 90+ day abstinence tier.
He takes a definitive test on March 23rd, which is negative.
Then, on March 26th, he takes a presumptive test, which detects opioids.
The service provider then orders a definitive test, also on March 26th, because of the results of presumptive test.
However, this definitive test will not be covered, because it is the second definitive test in a 7-day period, and David's abstinence tier just reset to 0-30 days, which only allows one definitive test in a 7-day period.

If the biller is working with the now-stale state that David's abstinence tier is 90+ days, the new definitive test appears covered, since it is the second test in a 90-day period.
In order to evaluate coverage correctly, the biller needs current information from the lab-order team, and from the clinical team that tracks abstinence tiers.
Therefore, if we agentify the workflow, with a LabOrder agent, a Billing agent, and a Clinical agent, we also need a meta-agent that can see the global state in order to make a final decision.
See \cref{tab:udt-agents} for a list of agents and their slices of state, in this agentified workflow.

Agentic Redux as presented in \cref{sec:intro-redux} is not sufficient to model UDT Compliance.
Both medical best practices and AHCCCS regulations mandate that the clinical response to a patient's relapse be a human decision point.
Therefore, the original Agentic Redux presented in \cref{fig:agentic-redux-architecture} is augmented with a Counselor Queue, that allows for a human in the loop, as shown in \cref{fig:udt-agentic-redux}.

The ontology of the UDT Compliance problem domain is available in the supporting code repository, as is fully working code for a UDT Compliance domain module, and Agentic Redux augmented with a Counselor Queue.
The independent continuant of the UDT Compliance ontology is the Patient.
The roles are as shown in \cref{tab:udt-agents}: LabOrder, Billing, and Clinical.
\begin{table}[t]
\centering
\caption{Sub-agent roles in the UDT compliance domain.}
\label{tab:udt-agents}
\begin{tabular}{@{}llp{5.8cm}@{}}
\toprule
\textbf{Role} & \textbf{Agent} & \textbf{Visible State Slice} \\
\midrule
\textsf{LabOrder} & Lab Order Agent & Clinical rationale \\
\hline
\textsf{Billing} & Billing Agent & Claim history \\
\hline
\textsf{Clinical} & Clinical Analysis Agent & Abstinence tier tracking  \\
\bottomrule
\end{tabular}
\end{table}
\begin{figure}[t]
     \centering
     \includegraphics[width=\linewidth]{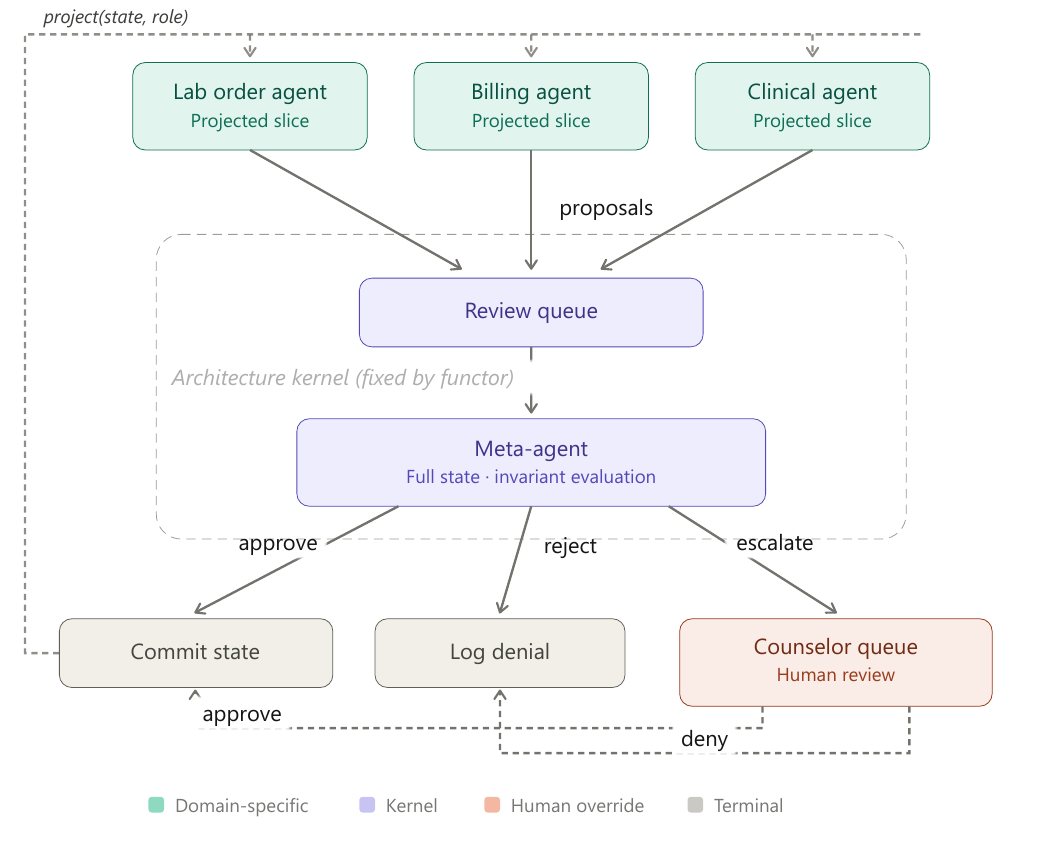}
     \caption{Agentic Redux expanded with a Counselor Queue for the UDT Compliance probem domain.}
     \label{fig:udt-agentic-redux}
   \end{figure}
\subsection{Security Vulnerability Disclosure}
\label{sec:security}
The Security Vulnerability Disclosure ontology is motivated by a bottleneck I first heard described
by Carlini in a video of a talk he gave at the [un]prompted security conference\cite{carlini-unprompted},
since explained in detail by Lynch\cite{lynch_explain_carlini},
and supported by an announcement from Anthropic Red\cite{anthropic_zero_days}.
The Anthropic Red announcement is particularly interesting for this paper,
because it shows that Anthropic Red has already converged on a methodology that is Agentic-Redux shaped.
The Anthropic Red methodology is less formal than Agentic Redux, and, at least to my mind, would benefit
from typed architecture that is correct by construction.
This section of the paper is my response to a call Carlini made at the end of his talk,
when he asked for assistance to get through the upcoming security crisis.

As with the UDT Compliance ontology, fully-tested source code for running
the Security Vulnerabiity Disclosure domain module on the Agentic Redux kernel
is available in the supporting code repository.
\subsubsection{Security Vulnerability Disclosure Ontology}
\label{sec:security-vulnerability-disclosure-ontology}
Following the method of Ontology-First Agent Design, I ontologized the Security Vulnerability Disclosure
problem domain, and then directed an LLM to derive roles that LLM agents and a human-in-the-loop would fill.
The scope of the problem domain was from just after the validation of a vulnerability to its eventual disclosure
to the organization owning the component containing the vulnerability.

While the UDT Compliance ontology has one independent continuant, the Patient,
the Security Vulnerability Disclosure ontology has two independent continuants: the \emph{Vulnerability},
and the \emph{Receiving Organization}. 
The Vulnerability is the primary entity that participates in
the disclosure adjudication pipeline; it is analogous to the Patient in the UDT Compliance ontology.
The Receiving Organization has no UDT Compliance analogue.
For purposes of this ontology (and this paper), it is primarily a source of signals that processes
in the disclosure adjudication pipeline wait for before proceeding.
If one were to ontologize the work to be done from disclosure to resolution,
the Receiving Organization would play a larger role.

Security Vulnerability Disclosure is a more complex problem domain than UDT Compliance,
as can be seen from the larger number of relations a Vulnerability can have with
Qualities and with Information Content Entities, see \cref{fig:vulnerability-relations}.
The differences are also visible if you compare the happy path of the independent
continuant's journey through each ontology, as shown in \cref{fig:journey-comparison}.

Despite the differences, \emph{the Agentic Redux architectures auto-derived from the two ontologies
are almost isomorphic}. (See \cref{fig:udt-agentic-redux} for the UDT Compliance architecture,
and \cref{fig:disclosure-agentic-redux} for the Security Vulnerability Disclosure architecture.)
The only difference is that there are three proposing subagents in UDT Compliance,
and four proposing subagents in Security Vulnerability Disclosure.
This is because the adjudication pipeline in both problem domains is essentially identical:
specalized team members with incomplete information propose actions to a decider with full
knowledge of the global state. Agentic Redux formalizes this shape to give it a semantic guarantee of safety.

An interesting feature of Security Vulnerability Disclosure is that, according to published reports,
Anthropic Red is already using an Agentic-Redux-shaped one-agent architecture to assist with 
the adjudication pipeline. Claude is positioned as both proposer and decider:
``I suggest this new action and here's why,'' ``Ok, let me decide based on everything else going on.''
This situation may allow Agentic Redux to offer \emph{complexity-theoretic} advantages (speedup),
not just safety.
Let's discuss complexity theory now.

\begin{figure}[t]
     \centering
     \includegraphics[width=\linewidth]{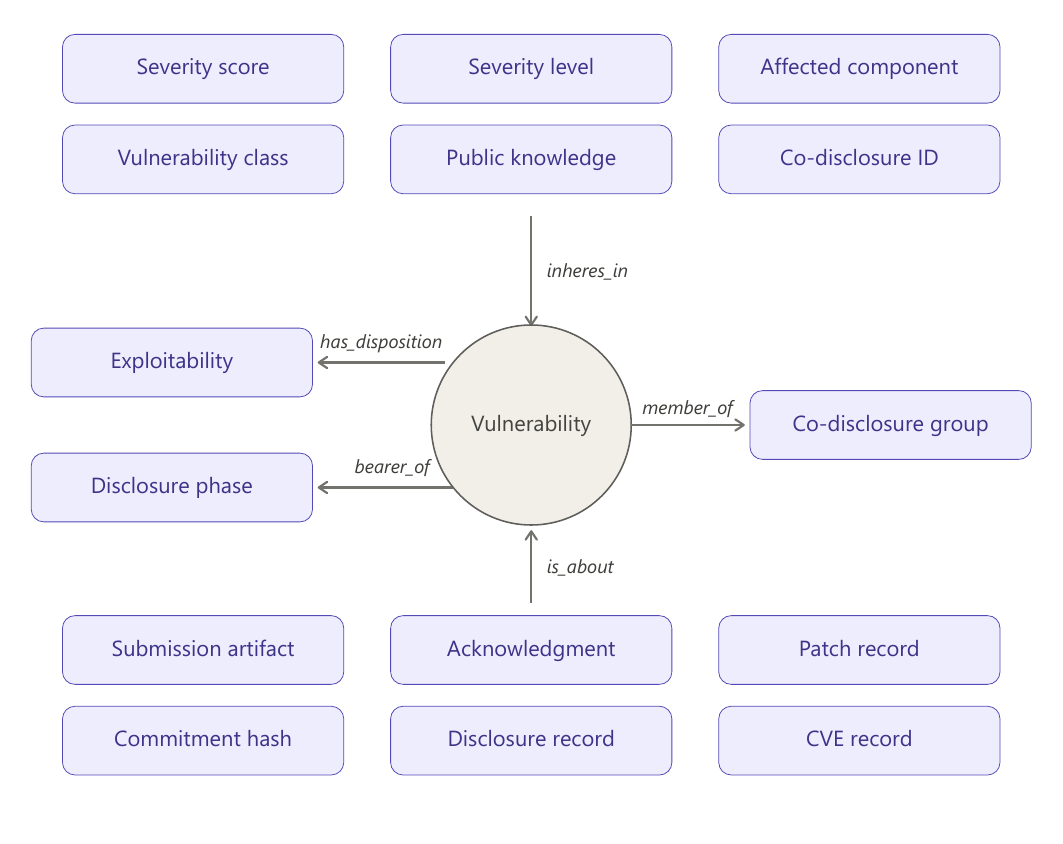}
     \caption{The independent continuant Vulnerability and its relations, in the Security Vulnerability Disclosure ontology.}
     \label{fig:vulnerability-relations}
   \end{figure}
\begin{figure}[t]
     \centering
     \includegraphics[width=\linewidth]{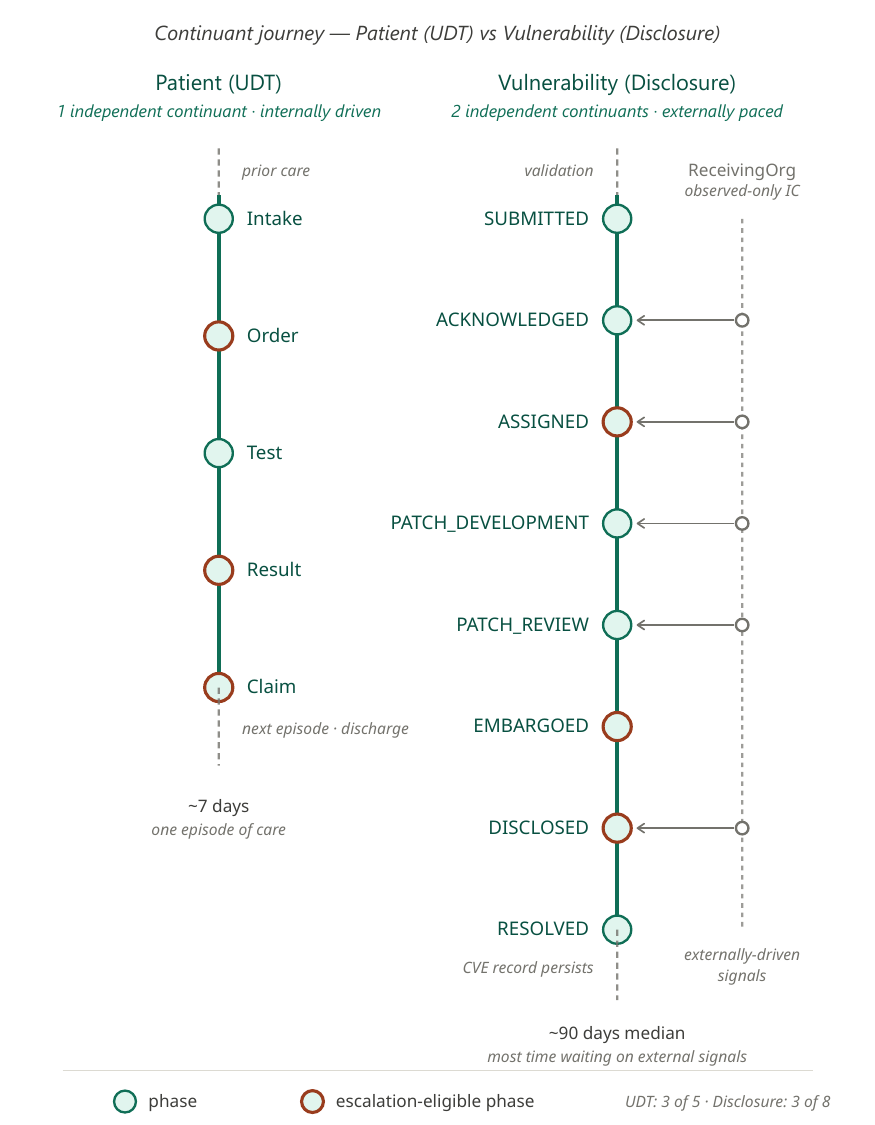}
     \caption{A comparison between the happy paths of the independent continuant's journey in the UDT Compliance ontology
     and the Security Vulnerability Disclosure ontology.}
     \label{fig:journey-comparison}
   \end{figure}
\begin{figure}[t]
     \centering
     \includegraphics[width=\linewidth]{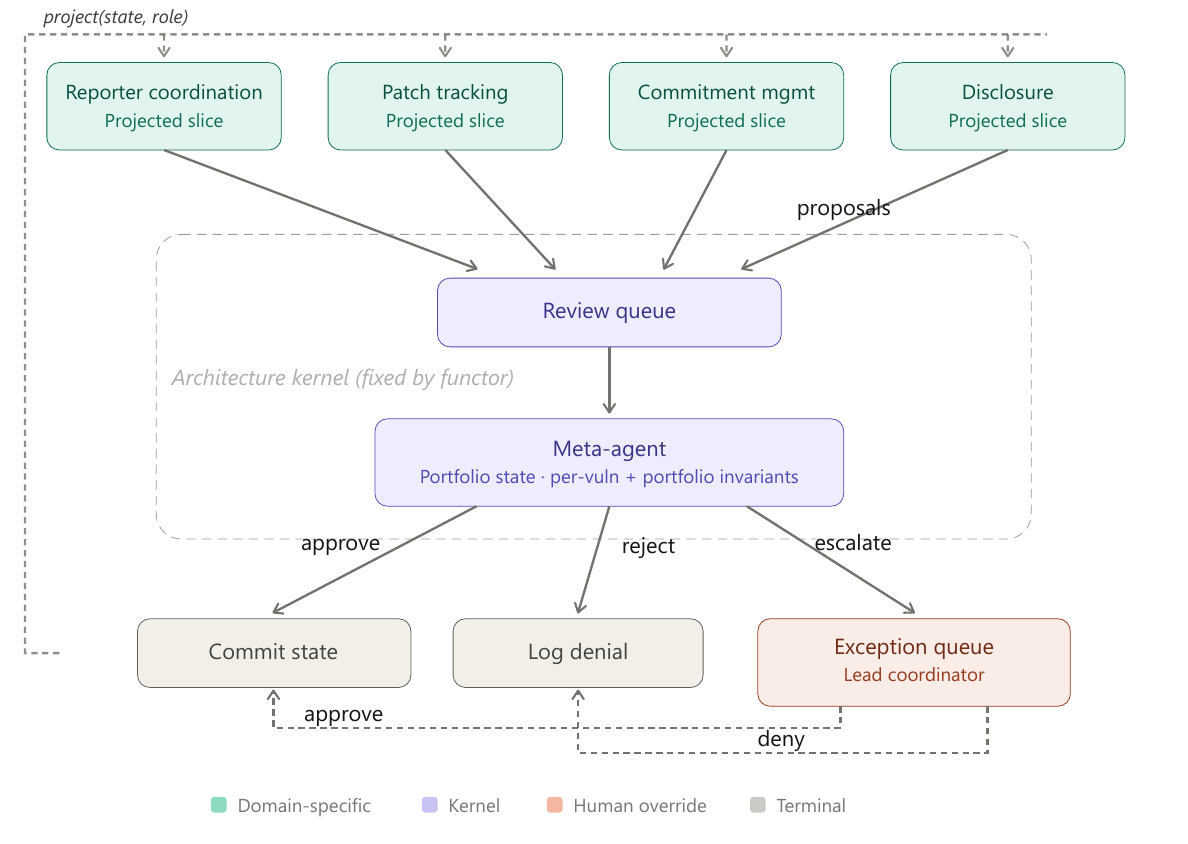}
     \caption{The Agentic Redux architecture for the Security Vulnerability Disclosure ontology.}
     \label{fig:disclosure-agentic-redux}
   \end{figure}
\subsubsection{Checkability Over Correctness}
\label{sec:checkability_over_correctness}
While the intuition that it is easier to verify than it is to create
has been part of computer science since at least the P/NP Problem,
I believe \emph{Checkability Over Correctness} was coined by Allan as a design principle
of the programming language Vera\cite{vera_programming_language}.
Applied to the context of Agentic Redux, the asymmetry between the potential complexity of subagents
and the simplicity of the meta-agent is good design.

The subagents have the hard job of proposing the correct next step for the global state.
This might involve LLM calls, tool calls, and the handling of ambiguity.
The meta-agent, by contrast, runs a simple check against a list of must-haves (the invariants).
If everything on the list checks green (all invariants are preserved), the meta-agent transitions the global state.
Otherwise, the meta-agent denies the proposal to transition.

While Agentic Redux ensures safety for every adjudicable problem domain,
in the case of Security Vulnerability Disclosure, \emph{the simplicity of the meta-agent may also provide speedup}.
Instead of the current reported Anthropic process, in which an LLM is on both sides of
the decision function (as both proposer and decider), under Agentic Redux, only the proposer is slow to act.
This time savings may add up when managing multiple vulnerabilities in the same co-disclosure group,
and potentially thousands of vulnerabilities overall.

For reasons of space, I defer machinery for liveness and 
complexity until \cite{semantic-correctness-2}.
Before leaving the topic, though, it's worth discussing an ontological difference
between Security Vulnerability \emph{Verification} and Security Vulnerability Disclosure.

I also ontologized Security Vulnerability Verification, the phase immediately before Security Vulnerability Disclosure,
in which a security finding is evaluated to determine whether it is a genuine vulnerability.
(I don't include that ontology in this paper, but it's available in the supporting code repository.)
The Agentic Redux architecture for Security Vulnerability Verification is literally
isomorphic to that of Security Vulnerability Disclosure---same diagram, different labels.
Yet, while Agentic Redux would provide a guarantee of safety to Security Vulnerability Verification, it may not provide
speedup. This is due to \emph{the ontological features of the human-in-the-loop queue}
that make vulnerability verification different from either UDT Compliance or Security Vulnerability Disclosure.

In both UDT Compliance and Security Vulnerability Disclosure, \emph{the human-in-the-loop queue is an exception handler}.
On a happy path, the agents may do all the work and never need human intervention.
With Security Vulnerability Verification, on the other hand, the verifier is a human,
so the happy path always includes the human-in-the-loop, and the human expert is responsible
for performing the verification. \emph{For vulnerability verification, checkability is hard},
so it seems difficult to gain speedup by leveraging Checkability Over Correctness.
\subsubsection{Security Agents and Invariants}
\label{sec:security-agents-invariants}
In the UDT Compliance ontology, all invariants were dependent only on internal information.
In the Security Vulnerability Disclosure ontology, some invariants depend on internal information,
while others depend on both internal and external information.
The external information arrives in the form of observed signals from the Receiving Organization,
that the disclosing organization records as data of the global state.

\texttt{MAINTAINER\_RATE\_LIMIT} is an example of an invariant of the first type.
It's an internal limit to prevent the overburdening of a software maintainer.
The meta-agent checks the current number of disclosures against an internal cap,
and denies the disclosure if the cap would be exceeded.

\texttt{CVE\_ATTRIBUTION\_CONSISTENCY} is an example of an invariant of the second type.
The system observes and records an external signal, assigning a ID number to a CVE.
The meta-agent then enforces that any proposal referring to that CVE must
use the externally-assigned ID number for that CVE.
The full list of invariants appears in \cref{tab:cvd-invariants}.

The subagents are very similar to the subagents of UDT Compliance.
Each works on, and makes proposals about, its own disjoint slice of state.
For example, the \emph{Patch Tracking Agent} updates patch information based on observations
made about the behavior of the Receiving Organization.
The Patch Tracking Agent's slice of state is limited to a single vulnerability,
so disclosure information cannot leak to other potential Receiving Organizations.
For a description of all four agents, see \cref{tab:disclosure-subagents}.
\begin{table}[htbp]
\centering
\caption{Domain invariants derived from the Security Vulnerability Disclosure ontology, 
grouped by the provenance of the state they read. 
The first group adjudicates against state populated entirely by the discovering organization's sub-agent actions; 
the second additionally reads recorded state populated by observation of external facts 
(CVE identifiers assigned by an external numbering authority, receiving-organization patch releases, public-knowledge transitions). 
Both groups are state predicates evaluated at the adjudication instant.}
\label{tab:cvd-invariants}
\footnotesize
\begin{tabularx}{\textwidth}{@{}lX@{}}
\toprule
\textbf{Invariant} & \textbf{Requirement} \\
\midrule

\multicolumn{2}{@{}l}{\textit{Adjudicable from discovering-organization state alone.}} \\
\midrule

\texttt{MAINTAINER\_RATE\_LIMIT}
& No more than $N$ reports to the same subsystem maintainer (within the same receiving organization) within a configured sliding window. The count aggregates across the portfolio and protects receiving-organization capacity. \\
\addlinespace

\texttt{COMMITMENT\_REGISTRY\_CONSISTENCY}
& The commitment registry is well-formed at every adjudication point: every published commitment hash corresponds to exactly one (artifact, vulnerability) pair; every resolution proposal targets an existing committed hash; and no resolution proposal targets a hash already resolved. The eventual-resolution obligation is deferred to Paper~2. \\
\addlinespace

\texttt{DISCLOSURE\_TIMELINE\_CONSISTENCY}
& When vulnerabilities are co-disclosed (members of the same \texttt{CoDisclosureGroup}), they must be at compatible disclosure phases at the planned joint disclosure timestamp. A disclosure-execution proposal for a group member is valid only if every other member is also at \texttt{EMBARGOED} (with disclosure preconditions met) or already \texttt{DISCLOSED} for the same group. \\
\addlinespace

\texttt{RECEIVING\_ORG\_CONFINEMENT}
& When the discovering organization is coordinating disclosure with multiple receiving organizations concurrently for distinct vulnerabilities, the patch-tracking role's slice for vulnerability $A$ in receiving organization $X$ must not contain details for vulnerability $B$ in receiving organization $Y$. State slices for concurrent disclosures with disjoint receiving organizations must themselves be disjoint. \\

\midrule

\multicolumn{2}{@{}l}{\textit{Adjudicable from recorded state, including observed external facts.}} \\
\midrule

\texttt{EMBARGO\_ENFORCEMENT}
& Vulnerability details cannot be publicly disclosed before the later of (a)~the maintainer's patch being released, or (b)~the embargo deadline elapsing. The specific embargo regime is configured via \texttt{EmbargoSpec}. Bypassed only when \texttt{PublicKnowledgeStatus} is \texttt{publicly-known}. \\
\addlinespace

\texttt{PATCH\_BEFORE\_DISCLOSURE}
& A vulnerability cannot be disclosed until a robust fix has been developed and released, except for vulnerabilities whose \texttt{PublicKnowledgeStatus} is \texttt{publicly-known} (fix released immediately upon availability) or whose embargo deadline has elapsed without a fix. \\
\addlinespace

\texttt{CVE\_ATTRIBUTION\_CONSISTENCY}
& If a CVE identifier is assigned, the disclosure record must reference the same identifier; if a CVE record already exists from intra-organizational deduplication, the disclosure must use that identifier rather than requesting a new one. \\
\addlinespace

\texttt{CVE\_ATTRIBUTION\_UNIQUENESS}
& Across the portfolio, each CVE identifier binds to exactly one vulnerability. No two distinct vulnerabilities share the same CVE attribution. \\
\addlinespace

\texttt{DISCLOSURE\_COMPLETENESS}
& The disclosure record at publication time must include (a)~a link to the released patch (or a statement of why the patch is unavailable for public-knowledge or embargo-elapsed cases), (b)~all resolved commitment hashes for the vulnerability, (c)~the assigned CVE identifier if one exists, and (d)~attribution. \\

\bottomrule
\end{tabularx}
\end{table}
\begin{table}[h]
\centering
\begin{tabular}{>{\raggedright\arraybackslash}p{0.28\linewidth} p{0.66\linewidth}}
\hline
\textbf{Agent} & \textbf{Work} \\
\hline
Reporter Coordination Agent
& Submits validated vulnerability reports to receiving organizations, records acknowledgment receipts, and proposes embargo extensions. Operates over the outbound submission queue and per-(receiving-org, maintainer) rate-limit status without portfolio-wide visibility. \\
\addlinespace
Patch Tracking Agent
& Updates the recorded \texttt{PatchRecord} status from externally-observable receiving-side signals (commits, release notes, advisories) for the assigned vulnerability and proposes severity revisions originating from the receiving organization. Slice is confined to one vulnerability so concurrent disclosures cannot leak across receiving organizations. \\
\addlinespace
Commitment Management Agent
& Publishes cryptographic commitment hashes for held artifacts (report, PoC, writeup) and proposes their resolution at disclosure time. Sees only the assigned vulnerability's hashes; the meta-agent checks registry-wide uniqueness and resolution targeting. \\
\addlinespace
Disclosure Agent
& Proposes co-disclosure groupings, asserts disclosure readiness, and executes public disclosure once the meta-agent clears the embargo, patch, completeness, CVE-attribution, and timeline-consistency invariants. Holds group-peer visibility to form proposals while leaving timeline adjudication centralized. \\
\hline
\end{tabular}
\caption{Subagents derived from the BFO-grounded security vulnerability disclosure ontology, one per \texttt{Role(Agent)} entry.}
\label{tab:disclosure-subagents}
\end{table}
%
\section{Overview of Formal Proofs}
\label{sec:overview-formal-proofs}

\subsection{Proofs About Agentic Redux}
\label{sec:proofs-about-agentic-redux}
Intuitively, a proof that Agentic Redux always preserves invariants might be: As shown in \cref{fig:agentic-redux-architecture}, global state cannot change unless the meta-agent changes it, and the meta-agent will only commit to a new global state if that new state preserves invariants.
While straightforward, it takes a fair amount of machinery to prove formally.
I have deferred the formal work to Appendix \cref{app:sec:calculus}, and I will discuss key points here.

To reason clearly about nondeterministic executions, the proofs use a standard technique of converting nondeterministic executions into deterministic executions that receive information from an oracle $\omega$.
You can think of $\omega$ as a possibly infinite sequence of data that contains each LLM response and tool use response, in order, received by the agents in a now-deterministic execution.
The theorem formalizing, ``Agentic Redux always preserves invariants,'' is as follows.
\paragraph{Theorem 1 (Invariant Preservation).}
\emph{For every $D : \Adjudicable$, every client program~$e_0$
well-typed over $\AgRedux(D)$, every initial framework state~$s_0$
of type $\FrameworkState_D$ with
$s_0.\mathsf{domain} \models \invariants_D$, and every oracle trace
$\omega \in \Oracle$: every reachable domain state
in~$\mathsf{Exec}(e_0, s_0, \omega)$ satisfies $\invariants_D$.}

The theorem statement informally translates to, ``For every domain module with the Auditability Fingerprint, 
if the initial state of the system preserves invariants, then every global state reachable through any nondeterministic execution 
of an Agentic Redux system preserves invariants.''
This theorem provides the semantic safety promised in the Introduction of the paper.
The proof of Theorem 1 appears in \cref{app:proof-invariant-safety}.

The guarantee of Theorem 1 is stronger than the Redux slogan, ``Smart container, dumb components.''
Even if the domain is buggy or malicious, the meta-agent of Agentic Redux ensures invariants are preserved.
\emph{As long as the invariants capture real-world policies, the rest of the components can be evil, not just dumb}.

With the Invariant Preservation theorem, we can now formalize two claims we have been making informally.
The first is \emph{Agent Information Confinement}: agents, at all steps of an execution, can only see, and only affect, their assigned slice of state.
That is formalized as Proposition 1, and proved in \cref{app:proof-confinement}.
The second claim is: \emph{Write Skew executions cannot occur under Agentic Redux}.
A Write Skew execution produces a global state that violates invariant preservation, which, by Theorem 1, Agentic Redux does not permit.
This argument is formalized as Corollary 1 in \cref{app:proof-write-skew}.

Other than preservation of invariants, the most important property we need from Agentic Redux is \emph{linear auditability}.
We obtain a formal proof that Agentic Redux is linearly auditable from Theorem 2 (Audit Log Integrity).
The theorem statement refers to several technical definitions, so I defer both the theorem statement and its proof to \cref{app:proof-audit-integrity}.
Informally, the Audit Log Integrity theorem asserts, ``If the system changes state, the action that caused the change is logged with the tag APPROVED; 
if a proposed state change is declined, either by REJECT or ESCALATE, the decision and the reason for the decision are both logged.
The log grows by exactly one entry per adjudication, from no other computation step, and entries are never modified once written.''

The final proof about Agentic Redux in the Appendix is Theorem 3 (Type Safety), stated and proved in \cref{app:proof-type-safety}.
We don't need Type Safety for any results in this paper.
However, Type Safety demonstrates that the Agentic Redux Calculus is well-behaved from the perspective of a programming language theorist.
It is also a property I intend to build on when considering other agent architectures in future work.
\subsection{Extending Proofs to the Counselor Queue}
\label{sec:proofs-counselor-queue}
The domain examples of \cref{sec:domains} include a Counselor Queue:
a mechanism by which an escalated proposal is resolved by a human
in the loop (\cref{fig:udt-agentic-redux}). When the meta-agent
encounters an invariant failure whose domain-declared outcome is
$\escalateRes$, the proposal is deposited into a review queue to
await a counselor's decision.

For purposes of a formal proof about all executions, the counselor has full policy authority.
The counselor, i.e., the supervising human in the loop, can change the global state in a way that violates invariants.
This may cause the system to behave badly, and it's the harness's job to clean things up, which is critically important in real life, but out of the scope of this paper.

Please note: \emph{the counselor has policy authority, but is type-theoretically constrained}.
The counselor cannot ``do anything,'' but must write a decision to the log, using the same consistent format as the meta-agent.
The main difference is that the counselor can violate invariants, while the meta-agent cannot.
To use the language of distributed systems, counselor decisions are not Byzantine failures.

The proof machinery assumes that, once the meta-agent has escalated a decision to the counselor, 
\emph{all activity stops until the counselor makes a decision}, to guarantee that the counselor decides with respect to the current global state.
That is a strong requirement, but it accurately models the real-world behavior of the problem domains considered in this paper.
In UDT Compliance, if the supervising MD needs to make a clinical decision, everyone waits for the decision, then responds accordingly.
Similarly, if a decision about a security vulnerability escalates to a human expert, 
work on that vulnerability report pauses until the expert makes a decision.
 
Let's call a state change directed by the counselor a \emph{Safe Counselor Commit} if the state change preserves all invariants.
After defining additional machinery (for example, now two writers can write to the log, not just one),
it's possible to show by induction that the system preserves invariants and maintains audit integrity, 
if every counselor commit is a Safe Counselor Commit.
Intuitively, if all invariants hold at time step 0, each commit of the meta-agent or the counselor preserves invariants, 
and each decision by the meta-agent or the counselor is appended to the log,
then every state of the system must preserve invariants, and the log is linearly auditable, as before.

Formal proofs of Counselor Queue extensions of theorems for Invariant Preservation, 
Audit Log Integrity, and Type Safety appear in \cref{app:counselor-base-preservation,app:proof-audit-integrity-ext}.
%
\section{Conclusion and Future Work}
\label{sec:conclusion-future-work}
This paper has focused on safety properties to support problem domains that require linear auditability.
It introduced the LLM agent architecture Agentic Redux, which contains a meta-agent
that enforces that every state change must respect global invariants, \emph{i.e.}, respect safety properties.
It also presented a procedure to semi-automatically design agents that solve a problem of interest,
in which a human expert ontologizes the problem domain,
then uses an LLM to derive agents that perform the roles needed to solve problems in that domain.

The followup paper\cite{semantic-correctness-2} will consider liveness properties,
and will introduce complexity measures to compare how different coding agent strategies perform on coding tasks.
The goal is to develop a framework with which it is possible to compare different programs run on the same
agent architecture, and the same program run on different agent architectures.
%
\section*{Acknowledgments}
\label{sec:acknowledgements}
I am grateful to my Thistleseeds cofounder Yuri Downing for raising the funds that made this research possible; 
and to Soma Chaudhuri, who taught me how to reason about distributed systems.

\FloatBarrier
\bibliographystyle{plainnat}
\bibliography{references}

\appendix
%
\section{The Agentic Redux Calculus}
\label{app:sec:calculus}

\emph{Note about the proofs in this Appendix:} The text in the main body of the paper
is 99\% text I wrote myself. The text in the Appendix, though, is about 10\%
mine and 90\% LLM-generated. 
I'll have more to say about the methodology soon.
For purposes of this paper: the prompts I provided were extensive, and the process was
iterative, not a one-shot. I caught errors in earlier versions of some proofs,
but the ones here now look good to me. I believe that these are the first LLM-enabled proofs
of program correctness written in the typed lambda calculus. 
Perhaps this approach can eventually become
a form of agentic self-verification.

\subsection{Preliminaries}
\label{sec:preliminaries}

We introduce the typed lambda calculus in which \AR{} is defined as
an architecture functor. Following the three-entity structure of
\cref{sec:intro-domains-functors}, a running system is obtained by
composing three pieces: a \emph{domain module} providing a conforming
domain signature, the \AR{} \emph{architecture functor} consuming
that signature, and a \emph{client program} that instantiates the
composition and sequences its operations. This section fixes the
calculus-level vocabulary in which each of these pieces is expressed
and in which the theorems of \cref{sec:metatheory} are stated. The
Role Derivation and Invariant Derivation procedures of
\cref{sec:ontology-to-agents} are the inputs to the formalisation; in
particular, the notion of a \emph{state footprint} introduced there
becomes a first-class object in the calculus.

This section fixes the vocabulary needed to state the five theorems of
\cref{sec:metatheory}, one proposition, and one corollary. No theorem
is proved here.

\subsubsection{State Signatures and Footprints}
\label{sec:state-signatures}

Procedure~\ref{proc:role-derivation} treats a state footprint as ``a
list of state components a process reads or writes.'' We make this
precise by fixing, per domain, a finite set of atomic state
components.

\begin{definition}[State signature]
\label{def:state-signature}
A \emph{state signature} for a domain~$D$ is a pair
$\Sigma_D = (\Fields_D, \tau_D)$ where:
\begin{itemize}
  \item $\Fields_D$ is a finite set of \emph{field names}, one per
        BFO-atomic state component declared by the ontology (a
        quality of the independent continuant, a role binding, an
        information content entity bound to the continuant, etc.);
  \item $\tau_D : \Fields_D \to \mathsf{Type}$ assigns a type to each
        field name.
\end{itemize}
The \emph{state type} of $D$ is the dependent record type
\[
  \State_D \;\triangleq\; \prod_{f \in \Fields_D} \tau_D(f).
\]
We write $s.f$ for the component of $s : \State_D$ at field~$f$.
\end{definition}

\begin{remark}[Atomicity]
\label{rem:bfo-atomicity}
The granularity of $\Fields_D$ is set by the ontology, not by the
representation chosen in code. In the UDT compliance domain, the BFO
ontology distinguishes the determinable \texttt{consecutiveDays} from
the derived determinable \texttt{tier}, so both appear as separate
elements of $\Fields_{\mathrm{UDT}}$ even though in code they are
siblings under an \texttt{abstinence} record. Record-level groupings
in the implementation are an encoding convenience; the state
signature is flat.
\end{remark}

\begin{definition}[Footprint]
\label{def:footprint}
A \emph{footprint} in~$D$ is a subset $F \subseteq \Fields_D$. The
\emph{restriction} of $s : \State_D$ to $F$ is the sub-record
\[
  s \restr F \;\triangleq\; (s.f)_{f \in F}
  \;:\; \prod_{f \in F} \tau_D(f).
\]
\end{definition}

Footprints form a lattice under subset ordering, with union and
intersection; both operations are used in
\cref{proc:role-derivation,proc:invariant-derivation}.

\subsubsection{Derived Footprints: Processes, Roles, Invariants}
\label{sec:derived-footprints}

\cref{proc:role-derivation,proc:invariant-derivation} assign
footprints to processes, roles, and invariants. We lift those
assignments into the calculus.

\begin{remark}[Whose roles?]
\label{rem:whose-roles}
BFO treats roles uniformly as qualities borne by independent
continuants. In the ontologies of \cref{sec:ontology-to-agents},
roles appear in two places:
\begin{itemize}
  \item \emph{Roles of the domain's independent continuant.} The
        Vulnerability bears a \texttt{PipelinePhase} role with
        values ranging over \texttt{DISCOVERED}, \texttt{TRIAGED},
        \ldots, \texttt{DISCLOSED}; the Patient bears a
        \texttt{ProgramPhase} role. These are qualities of the
        continuant that evolve as it moves through its lifecycle.
        Formally, they live in $\Fields_D$ as ordinary state
        fields, typically with an enumerated type. They do not
        mandate the existence of an agent.
  \item \emph{Roles of workers.} Procedure~\ref{proc:role-derivation}
        introduces \emph{workers} --- new independent continuants
        bearing roles such as \texttt{StaticAnalysisRole},
        \texttt{LabOrderRole}, \texttt{TriageRole}. These roles
        come with footprints and process assignments; they are
        filled by agents or by humans in the loop; they are the
        subject of Agent Confinement.
\end{itemize}
In this section, $\Roles_D$ refers to the second kind only: the
worker roles produced by Procedure~\ref{proc:role-derivation}. A
reference to a pipeline-phase-type field in an invariant footprint
(e.g., ``Disclosure requires \texttt{PipelinePhase = EMBARGOED}'')
is a reference to a $\Fields_D$ element, not to a $\Roles_D$
element.
\end{remark}

\begin{definition}[Process footprint]
\label{def:process-footprint}
For each process $p$ declared in the ontology, $\Footprint(p)
\subseteq \Fields_D$ is the footprint specified by
Procedure~\ref{proc:role-derivation}: every field read or written
during~$p$.
\end{definition}

\begin{definition}[Agent role and role footprint]
\label{def:role-footprint}
An \emph{agent role} (or \emph{worker role}) $r \in \Roles_D$ is a
label drawn from a finite set fixed by the domain, together with an
assignment $r \mapsto \mathsf{procs}(r)$ of a set of processes
to~$r$. The \emph{role footprint} is
\[
  \Footprint(r) \;\triangleq\;
  \bigcup_{p \in \mathsf{procs}(r)} \Footprint(p).
\]
In the remainder of this section and in all subsequent theorems,
``role'' without qualification refers to an agent role in this
sense (cf.\ Remark~\ref{rem:whose-roles}).
\end{definition}

\begin{definition}[Slice]
\label{def:slice}
The \emph{slice} for role~$r$ is the restricted record type
\[
  \Slice_D(r) \;\triangleq\;
  \prod_{f \in \Footprint(r)} \tau_D(f),
\]
and the \emph{projection function} is
$\project_D(s, r) \triangleq s \restr \Footprint(r)$.
\end{definition}

\begin{definition}[Well-formed role assignment]
\label{def:wf-roles}
A domain~$D$ has a \emph{well-formed role assignment} iff for every
role~$r$ and every field $f \in \Footprint(r)$ there exists a process
$p \in \mathsf{procs}(r)$ with $f \in \Footprint(p)$.
\end{definition}

Well-formedness formalises the final check of
Procedure~\ref{proc:role-derivation} (``no role touches state its
processes do not need''). All theorems in \cref{sec:metatheory}
assume the role assignment is well-formed.

\begin{definition}[Invariant footprint]
\label{def:invariant-footprint}
Let $\invariants_D$ be the ordered list of invariants declared by
Procedure~\ref{proc:invariant-derivation}. For each invariant $\iota
\in \invariants_D$, $\Footprint(\iota) \subseteq \Fields_D$ is the
footprint declared in that procedure: the fields $\iota$ references.
\end{definition}

\begin{definition}[Local and cross-cutting invariants]
\label{def:local-crosscutting}
An invariant $\iota$ is \emph{local to role~$r$} iff
$\Footprint(\iota) \subseteq \Footprint(r)$. It is
\emph{cross-cutting} iff $\Footprint(\iota) \not\subseteq
\Footprint(r)$ for every $r \in \Roles_D$.
\end{definition}

This matches the classification of
\cref{sec:system-invariants}: an invariant is cross-cutting precisely
when no single role's slice suffices to evaluate it.

\begin{definition}[Coordination-requiring invariant set]
\label{def:coordination-requiring}
The invariant set $\invariants_D$ is \emph{coordination-requiring}
iff there exist a state $s \in \State_D$ with $s \models
\invariants_D$, actions $a_1, a_2 \in \Action_D$, and timestamps
$t_1, t_2 \in \Time$ such that
$\applyMutation_D(s, a_i, t_i) \models \invariants_D$ for
$i \in \{1, 2\}$, yet applying $a_1$ and $a_2$ independently
from~$s$ can yield a state violating some $\iota \in \invariants_D$.
\end{definition}

A coordination-requiring invariant set is the footprint-based
counterpart of the failure of \emph{I-confluence} in the sense of
Bailis \emph{et al.}~\cite{bailis2014coordination}: no coordination-free
execution of agents producing $a_1$ and $a_2$ from disjoint slices
can be guaranteed to preserve $\invariants_D$. The Write Skew
example of the Introduction is the canonical witness: the budget
invariant is coordination-requiring because independent
\$45{,}000 and \$60{,}000 proposals both preserve it from the
initial state yet their composition violates it. A domain whose
invariant set is coordination-requiring needs an architecture that
provides coordination; \AR{} is one such architecture.

\subsubsection{Domain Signatures}
\label{sec:domain-signatures}

The architecture functor consumes a weak domain signature
(\Adjudicable); the stronger signature (\Domain) carries the role
structure needed to state Agent Confinement.

\begin{definition}[Three-valued invariant result]
\label{def:inv-result}
An invariant returns a value in the sum type
\[
  \InvResult \;\triangleq\;
  \{\pass\} \;+\; \rejectRes(\mathsf{String}) \;+\; \escalateRes(\mathsf{String})
\]
signalling, respectively, that the proposed action is compatible
with the invariant; that it violates the invariant, with an
accompanying diagnostic message; or that it requires human review
before commit, with an accompanying explanation. We write
$\rejectRes(m)$ and $\escalateRes(m)$ for the payload-carrying
values.
\end{definition}

\begin{definition}[Signature \Adjudicable]
\label{def:adjudicable}
A structure~$D$ satisfies \Adjudicable{} iff it provides:
\begin{itemize}
  \item a state signature $\Sigma_D = (\Fields_D, \tau_D)$;
  \item a type $\Action_D : \Occ$ of proposable actions;
  \item an ordered list of invariants $\invariants_D = [\iota_1,
        \ldots, \iota_n]$ with each $\iota_i$ equipped with:
        \begin{itemize}
          \item an identifier $\mathrm{id}(\iota_i) \in
                \mathsf{InvariantId}$, distinct across the list,
          \item a footprint $\Footprint(\iota_i) \subseteq \Fields_D$,
          \item a predicate
                $\iota_i : \State_D \to \InvResult$;
        \end{itemize}
  \item a total pure function
        $\applyMutation_D : \State_D \times \Action_D \times \Time
        \to \State_D$.
\end{itemize}
\end{definition}

\begin{remark}[Invariants as state predicates]
\label{rem:state-predicates}
The predicate $\iota_i$ takes only a domain state; it has no
dependence on the action whose proposal triggered evaluation, nor on
any pre-state separate from the one it is checking, nor on the
oracle trace: the predicate is fully determined by its state
argument, as declared by the signature $\State_D \to \InvResult$.
Invariants are fixed by the domain module at time~$0$, before any
oracle is consulted, and the meta-agent cannot approve a state that
violates them regardless of what oracle string its decision
function reads. This is the
central semantic choice of the architecture: the meta-agent alone
bears responsibility for preservation of $\invariants_D$, and in
particular does not delegate that responsibility to the domain's
$\applyMutation_D$ function. A mutation that produces a state
violating some $\iota_i$ --- whether through a bug, a
misspecification, or adversarial input --- is rejected by the
post-check before any commit is performed. See
Definition~\ref{def:commit-ar} for the commit discipline that
realises this.
\end{remark}

\begin{definition}[Signature \Domain]
\label{def:domain-signature}
A structure~$D$ satisfies \Domain{} iff it satisfies \Adjudicable{}
and, in addition, provides:
\begin{itemize}
  \item a finite set of \emph{agent roles} $\Roles_D$ --- exactly the
        worker roles produced by
        Procedure~\ref{proc:role-derivation}, not the BFO roles of
        the domain's independent continuant (cf.\
        Remark~\ref{rem:whose-roles});
  \item a process assignment
        $\mathsf{procs} : \Roles_D \to \wp(\mathsf{Procs}_D)$,
        inducing a footprint $\Footprint(r)$ for each role;
  \item a projection function
        $\project_D : \State_D \times (r : \Roles_D) \to \Slice_D(r)$
        with $\project_D(s, r) = s \restr \Footprint(r)$;
\end{itemize}
subject to the well-formed role assignment condition
(Definition~\ref{def:wf-roles}).
\end{definition}

\begin{remark}[The footprint determines the slice]
\label{rem:footprint-determines-slice}
Under \Domain, the slice type $\Slice_D(r)$ is not a free parameter;
it is computed from $\Footprint(r)$, which is in turn computed from
$\mathsf{procs}(r)$. A reviewer can audit ``is this role seeing only
what it should?'' by reading the role's process assignment and
comparing against the footprint definition. This is the formal
content of Procedure~\ref{proc:role-derivation}.
\end{remark}

\subsubsection{Invariant Satisfaction}
\label{sec:invariants-sat}

\begin{definition}[Invariant satisfaction]
\label{def:invariant-satisfaction}
A state $s : \State_D$ \emph{satisfies the invariant pipeline}
of~$D$, written $s \models \invariants_D$, iff for every invariant
$\iota \in \invariants_D$, $\iota(s) = \pass$.
\end{definition}

The sequential, short-circuited evaluation of the kernel's pipeline
is captured by the commit discipline (Definition~\ref{def:commit-ar})
rather than by $\models$. The relation $\models$ records the model
condition; the first-non-$\pass$ ordering records which invariant
witnessed a failure.

\begin{definition}[Failure witness]
\label{def:failure-witness}
If $s \not\models \invariants_D$, the \emph{failure witness} is the
first invariant $\iota_k$ in list order for which
$\iota_k(s) \ne \pass$. The failure witness determines whether the
outcome is $\rejected$ or $\escalated$.
\end{definition}

\begin{remark}[Phase gating in lieu of action filtering]
\label{rem:phase-gating}
An invariant with an opinion only in certain domain-state shapes is
expressed by having its predicate early-return $\pass$ outside those
shapes. A typical pattern reads a phase-valued field from
$\Fields_D$ and returns $\pass$ unless that field takes a specified
value. This is a domain-authoring convention, not a kernel concept:
the kernel runs every invariant on every candidate, and phase-gated
predicates become no-ops in irrelevant phases. Any scoping that might
otherwise have been expressed as ``this invariant applies only to
these action types'' is recast as ``this invariant has an opinion
only in these state shapes.''
\end{remark}

\subsubsection{Oracles and Arbitration}
\label{sec:oracles}

Agents and the meta-agent may consult external computations (LLMs,
calculators, lookup tables, human counsellors at a queue) whose
responses are not determined by the domain state. The client program
may also schedule agents in different orders when their proposals
arrive simultaneously in real time. We absorb all such
non-determinism into a single oracle trace consulted by the
execution.

\begin{definition}[Oracle trace]
\label{def:oracle}
Fix an abstract countable alphabet $\mathbb{O}$ of oracle responses.
An \emph{oracle trace} is a finite or infinite sequence $\omega \in
\mathbb{O}^{*} \cup \mathbb{O}^{\omega}$. The space of oracle traces
is denoted $\Oracle$. A configuration that consumes an oracle
advances through~$\omega$ one element per consultation.
\end{definition}

Oracle consultations include: LLM inference outputs, random-number
outputs, counsellor-queue decisions, wall-clock timestamps, and the
arbitration order in which the coordination protocol elicits
proposals from agents. Nothing in $\omega$ is a function of the
domain state; any such dependency enters an agent via its slice or
the meta-agent via its state input.

\subsubsection{Agents}
\label{sec:agents}

\begin{definition}[Agent]
\label{def:agent}
Let $D : \Domain$ and $r \in \Roles_D$. An \emph{agent for role~$r$
over~$D$} is a computable partial function
\[
  \alpha_r \;:\; \Slice_D(r) \times \Oracle
             \rightharpoonup \Action_D(r)
\]
satisfying:
\begin{enumerate}
  \item \emph{Typed output.} When $\alpha_r(\sigma, \omega)$ is
        defined, its value is an action classified as belonging to
        role~$r$ (i.e.\ tagged with~$r$ in whatever discipline the
        domain uses).
  \item \emph{Slice-only state dependence.} The function takes no
        other inputs that depend on $\State_D$; in particular,
        $\alpha_r$ does not receive state from outside
        $\Footprint(r)$ and does not observe the slices of other
        agents.
  \item \emph{Oracle determinism.} For fixed $\sigma$ and~$\omega$,
        $\alpha_r(\sigma, \omega)$ is deterministic.
\end{enumerate}
The space of agents for role~$r$ is written $\Agent_D(r)$.
\end{definition}

\begin{remark}[Realisations]
\label{rem:agent-realisations}
Condition~2 constrains the \emph{interface}, not the
\emph{implementation}. An agent function may be realised by a
classical algorithm, by an LLM consulted with prompts constructed
from~$\sigma$, by a composition of pure tool calls over~$\sigma$, or
by a human typing decisions at a terminal. In each realisation, the
LLM response, tool output, or human decision appears as an element
consumed from~$\omega$. Synchronous and asynchronous implementations
are equally admissible; asynchronicity is an operational detail
invisible at the calculus level. A shared-oracle semantics is used
throughout (one~$\omega$ per execution, consulted by all agents and
the meta-agent); a per-role reformulation with projections
$\omega \restr r$ is equivalent.
\end{remark}

\begin{remark}[No peer-to-peer communication]
\label{rem:no-peer}
The type of $\alpha_r$ has no channel for messages from other
agents. The calculus has no primitive permitting Agent~A to observe
Agent~B's slice, to receive B's proposal, or to await B's completion.
The only indirect flow of information from~B to~A is through the
meta-agent: if B's proposal is committed and the commit touches a
field in $\Footprint(A) \cap \Footprint(B)$, then A's next
invocation will see the new value in its slice. This is mediated by
the store, not by a message channel. The hub-and-spoke topology of
\cref{fig:agentic-redux-architecture} is therefore not merely
stipulated but enforced by the agent type.
\end{remark}

\begin{remark}[Tools as pure functions of the slice]
\label{rem:tools}
A tool invoked by an agent that read a field $f \notin \Footprint(r)$
would inject state outside the slice into $\alpha_r$'s output,
violating Condition~2. At the calculus level this is ruled out
structurally: a well-typed $\alpha_r$ can only construct queries
from~$\sigma$ and consume responses from~$\omega$. A governance
layer such as Microsoft Agent Framework~\cite{microsoftagentframework}
provides operational enforcement of the same discipline at runtime.
\end{remark}

\subsubsection{The Meta-Agent}
\label{sec:meta-agent}

\begin{definition}[Meta-agent]
\label{def:meta-agent}
Let $D : \Adjudicable$. A \emph{meta-agent for~$D$} is a computable
partial function
\[
  \mu_D \;:\;
  \State_D \times \Action_D \times \Time \times \Oracle
  \rightharpoonup \Outcome_D,
\]
where $\Outcome_D \triangleq \Approval(\State_D) + \Rejection +
\Escalation(\PendingProposal_D)$. The meta-agent has full state
access and is the only component permitted to produce an
$\approved$-tagged outcome.
\end{definition}

The meta-agent's function is supplied by the architecture, not by the
domain. Agentic Redux fixes~$\mu_D$ to be the concrete adjudication
function defined in \cref{sec:coord-commit}. A different
architecture (e.g., an LLM-judged adjudicator) may supply a different
$\mu_D$ with the same type.

\begin{remark}[The meta-agent as a serialisation point]
\label{rem:serialisation-point}
The meta-agent is the architecture's \emph{serialisation point} in
the sense of classical concurrency control: the component through
which all state-changing proposals pass in a well-defined order.
Alternative architectures may instantiate the serialisation point
differently --- for example, by admitting a bounded window of
speculative parallel proposals that are validated against a
snapshot, or by layering partition-local serialisation points
beneath a global one. Such variants are beyond the scope of this
paper; we note only that the architecture signature of
\cref{def:agent-architecture} is rich enough to accommodate them.
\end{remark}

\subsubsection{Agent Population}
\label{sec:agent-population}

\begin{definition}[Agent population]
\label{def:agent-population}
An \emph{agent population} over $D : \Domain$ is a family
\[
  \AgentPop \;=\; (\alpha_r)_{r \in \Roles_D},
  \quad \alpha_r \in \Agent_D(r)
\]
containing exactly one agent per role.
\end{definition}

The ``given agents A, B, C'' template is made precise by fixing
$\AgentPop = (\alpha_A, \alpha_B, \alpha_C)$. An architecture
\emph{for} these agents is the data of
\cref{sec:architecture-definition}.

\subsubsection{Framework State}
\label{sec:framework-state}

The kernel's operational state is a framework envelope around the
domain state, carrying the audit log and the review queue.

\begin{definition}[Framework state]
\label{def:framework-state}
For $D : \Adjudicable$, the \emph{framework state} type is
\[
  \FrameworkState_D \;\triangleq\;
  \mathsf{EntityId} \times \State_D \times
  \mathsf{List}(\AuditEntry_D) \times \ReviewQueue_D.
\]
We write $s.\mathsf{domain}$, $s.\mathsf{audit}$, $s.\mathsf{queue}$
for the projections onto the domain state, the audit log, and the
review queue respectively.
\end{definition}

\begin{definition}[Audit entry]
\label{def:audit-entry}
An \emph{audit entry} records the outcome of an adjudication:
\begin{align*}
  \AuditEntry_D \;\triangleq\;
  &\mathsf{Id} \times \Time \times
   \{\approved, \rejected, \escalated\} \\
  &\times \Action_D \times
   \mathsf{Option}(\mathsf{InvariantId} \times \mathsf{String}).
\end{align*}
\end{definition}

\begin{definition}[Review queue and pending proposal]
\label{def:review-queue}
\[
  \PendingProposal_D \;\triangleq\;
  \mathsf{Id} \times \Action_D \times \mathsf{InvariantId}
  \times \mathsf{String} \times \Time,
\]
\[
  \ReviewQueue_D \;\triangleq\;
  \mathsf{List}(\PendingProposal_D).
\]
\end{definition}

The audit log and review queue evolve deterministically with every
adjudication: an approved commit appends one audit entry; a
rejection appends an audit entry and leaves the queue unchanged; an
escalation appends both an audit entry and a pending proposal.
Domain invariants $\invariants_D$ are statements about
$s.\mathsf{domain}$, not about $s.\mathsf{audit}$ or
$s.\mathsf{queue}$; integrity of the audit log and queue is a
separate kernel-level property.

\subsubsection{Coordination Protocol and Commit Discipline}
\label{sec:coord-commit}

An architecture is individuated by two independent pieces of data:
who runs when, and how proposals become committed transitions.

\begin{definition}[Coordination protocol]
\label{def:coord}
A \emph{coordination protocol} is a computable partial function
\[
  \Coord \;:\;
  \FrameworkState_D \times \AgentPop \times \Oracle
  \rightharpoonup \Roles_D \times \Action_D.
\]
Given the current framework state, the
agent population, and the oracle trace, $\Coord$ selects a
role~$r$, invokes $\alpha_r$ on $\project_D(s.\mathsf{domain}, r)$,
and returns the pair $(r, a)$ where $a$ is the returned action.
\end{definition}

\begin{definition}[Commit discipline]
\label{def:commit}
A \emph{commit discipline} is a computable partial function
\[
  \Commit \;:\;
  \FrameworkState_D \times \Action_D \times \Time \times \Oracle
  \rightharpoonup \FrameworkState_D.
\]
Given the current framework state, a proposed action, a timestamp,
and the oracle trace, $\Commit$ produces the next framework state,
including any updates to the audit log and review queue.
\end{definition}

The separation is deliberate: $\Coord$ decides what is proposed and
by whom; $\Commit$ decides what happens to proposals. Agentic Redux,
and the extensions previewed for future work, differ in
$(\Coord, \Commit)$ while sharing the same agent population type.

\subsubsection{Agent Architecture}
\label{sec:architecture-definition}

\begin{definition}[Agent architecture]
\label{def:agent-architecture}
An \emph{agent architecture} over a domain $D : \Domain$ is a tuple
\[
  \Arch \;=\; (\AgentPop,\; \mu_D,\; \Coord,\; \Commit)
\]
consisting of an agent population, a meta-agent, a coordination
protocol, and a commit discipline.
\end{definition}

The ``given agents A, B, C, an architecture is Y'' template reads, in
this vocabulary: an architecture for $(\alpha_A, \alpha_B,
\alpha_C)$ is any tuple $(\AgentPop, \mu_D, \Coord, \Commit)$ with
$\AgentPop = (\alpha_A, \alpha_B, \alpha_C)$. Agentic Redux is the
specific instance of \cref{sec:agentic-redux-instance} below.

\subsubsection{Agentic Redux as an Instance}
\label{sec:agentic-redux-instance}

\begin{definition}[Agentic Redux coordination protocol]
\label{def:coord-ar}
The \emph{Agentic Redux coordination protocol} $\Coord^{\AR}$
operates as follows. At each step, given a framework state~$s$, an
agent population $\AgentPop$, and an oracle trace~$\omega$:
\begin{enumerate}
  \item consult~$\omega$ to select a role $r \in \Roles_D$ (the
        \emph{arbitration choice});
  \item compute the slice
        $\sigma = \project_D(s.\mathsf{domain}, r)$;
  \item invoke $\alpha_r(\sigma, \omega)$ to obtain an action~$a$;
  \item return $(r, a)$.
\end{enumerate}
At most one $(r, a)$ pair is in flight to the commit discipline at
any time. Simultaneity among real-world agents is absorbed into the
arbitration choice in Step~1.
\end{definition}

\begin{definition}[Agentic Redux commit discipline]
\label{def:commit-ar}
The \emph{Agentic Redux commit discipline} $\Commit^{\AR}$ is the
function that, on input $(s, a, t, \omega)$:
\begin{enumerate}
  \item computes the \emph{candidate domain state}
        $\hat{s} \;\triangleq\;
        \applyMutation_D(s.\mathsf{domain}, a, t)$;
  \item walks $\invariants_D$ in list order, evaluating each
        invariant on~$\hat{s}$;
  \item halts on the first invariant $\iota_k \in \invariants_D$
        whose evaluation on~$\hat{s}$ does not return~$\pass$:
        \begin{itemize}
          \item if $\iota_k(\hat{s}) = \rejectRes(m)$, appends to
                $s.\mathsf{audit}$ a single entry with tag
                $\rejected$, action~$a$, timestamp~$t$, and witness
                slot $\mathsf{Some}(\mathrm{id}(\iota_k), m)$, and
                returns the framework state with
                $\mathsf{domain}$ and $\mathsf{queue}$ components
                unchanged;
          \item if $\iota_k(\hat{s}) = \escalateRes(m)$, appends to
                $s.\mathsf{audit}$ a single entry with tag
                $\escalated$, action~$a$, timestamp~$t$, and
                witness slot
                $\mathsf{Some}(\mathrm{id}(\iota_k), m)$, and
                appends to $s.\mathsf{queue}$ a single pending
                proposal carrying action~$a$,
                $\mathrm{id}(\iota_k)$, $m$, and~$t$, returning the
                framework state with $\mathsf{domain}$ component
                unchanged;
        \end{itemize}
  \item if every invariant returns $\pass$ on~$\hat{s}$, appends to
        $s.\mathsf{audit}$ a single entry with tag $\approved$,
        action~$a$, timestamp~$t$, and witness slot
        $\mathsf{None}$, and returns the framework state with
        $\mathsf{domain}$ component~$\hat{s}$ and $\mathsf{queue}$
        component unchanged.
\end{enumerate}
Each branch appends exactly one audit entry; the approve and reject
branches leave the queue unchanged, and the escalate branch appends
exactly one pending proposal.
\end{definition}

\begin{remark}[Eager validation and strict serialisation]
\label{rem:eager-validation}
$\Commit^{\AR}$ produces a candidate domain state via
$\applyMutation_D$ and then evaluates every invariant on that
candidate \emph{before} committing it to the framework state (Step~4
of \cref{def:commit-ar}). In the terminology of concurrency control,
this is \emph{eager validation} combined with strict serialisation:
proposals are dispatched one at a time
(\cref{rem:seq-and-simultaneity}), invariants are checked
synchronously against the candidate, and the candidate is committed
only when every invariant has returned $\pass$. The isolation level
is serializable in the sense of Berenson \emph{et
al.}~\cite{berenson}: Write Skew (anomaly A5B) is structurally
excluded (\cref{sec:metatheory}, Corollary~1).
\end{remark}

\begin{definition}[The Agentic Redux architecture]
\label{def:ar-architecture}
For $D : \Domain$, and for any agent population $\AgentPop$
over~$D$, the \emph{Agentic Redux architecture over~$D$ with
population~$\AgentPop$} is the tuple
\[
  \AgRedux(D, \AgentPop) \;\triangleq\;
  (\AgentPop,\; \mu_D^{\AR},\; \Coord^{\AR},\; \Commit^{\AR}),
\]
where $\mu_D^{\AR}$ is the meta-agent whose behaviour is specified
by $\Commit^{\AR}$.
\end{definition}

When the population is clear from context we write simply
$\AgRedux(D)$.

\subsubsection{Configurations, Reduction, Executions}
\label{sec:executions}

The calculus's operational semantics is a deterministic small-step
reduction on configurations parameterised by an oracle trace. The
reduction relation is sequential: at most one step occurs per
transition, and there is no concurrent or parallel composition.

\begin{definition}[Configuration]
\label{def:configuration}
A \emph{configuration} is a triple
$\config{e}{s}{\omega}$ where:
\begin{itemize}
  \item $e$ is the client program term under reduction;
  \item $s : \FrameworkState_D$ is the current framework state;
  \item $\omega \in \Oracle$ is the remaining oracle trace.
\end{itemize}
The oracle position advances as the reduction consumes trace
elements.
\end{definition}

\begin{definition}[Well-typed configuration over $\AgRedux(D)$]
\label{def:wt-config}
A configuration $\config{e}{s}{\omega}$ is \emph{well-typed over}
$\AgRedux(D)$ iff:
\begin{enumerate}
  \item $e$ is closed and $\emptyset \vdash_D e : \tau$ for some
        type~$\tau$, per the typing rules of \cref{sec:typing};
  \item $s : \FrameworkState_D$ and
        $s.\mathsf{domain} : \State_D$.
\end{enumerate}
\end{definition}

\begin{remark}[Slice freshness via the reduction rules, not typing]
\label{rem:freshness-operational}
An earlier draft of Definition~\ref{def:wt-config} included a
third condition constraining slices passed to agents to be derived
from the current framework state. That content is enforced
structurally by \textsc{E-Coordinate}, which reads~$s$ from the
enclosing configuration at the moment the rule fires and computes
$\project_D(s.\mathsf{domain}, r)$ at that same moment, together
with the term grammar's lack of a value form for slices
(Remark~\ref{rem:slice-freshness}), which prevents any reduction
from introducing a captured slice for later invocation. The
dropped typing condition is therefore redundant with the
operational semantics.
\end{remark}

\begin{definition}[Reduction]
\label{def:reduction}
The \emph{reduction relation}
\[
  \config{e}{s}{\omega}
  \;\longrightarrow\;
  \config{e'}{s'}{\omega'}
\]
is a deterministic single-step relation on configurations. Its
defining rules appear in \cref{sec:reduction-rules}; the key
architecture rules are \textsc{E-Coordinate} (invoking $\Coord^{\AR}$
to elicit a proposal) and \textsc{E-Adjudicate} (invoking
$\Commit^{\AR}$ to produce an outcome). Each rule consumes zero or
more elements from~$\omega$ and advances the trace accordingly. We
write $\longrightarrow^{*}$ for the reflexive-transitive closure.
\end{definition}

\begin{definition}[Execution]
\label{def:execution}
An \emph{execution of~$e_0$ from initial framework state~$s_0$ with
oracle trace~$\omega$} is the unique (finite or infinite) sequence
\[
  \config{e_0}{s_0}{\omega}
  \;\longrightarrow\;
  \config{e_1}{s_1}{\omega_1}
  \;\longrightarrow\;
  \config{e_2}{s_2}{\omega_2}
  \;\longrightarrow\;
  \cdots
\]
generated by the reduction relation, where each $\omega_{i+1}$ is
the suffix of $\omega_i$ remaining after the reduction step consumed
its oracle reads. We write $\mathsf{Exec}(e_0, s_0, \omega)$ for
this sequence.
\end{definition}

\begin{remark}[Sequentiality absorbs simultaneity]
\label{rem:seq-and-simultaneity}
The reduction relation is a function of the configuration and the
oracle: each non-terminal configuration has exactly one successor.
Concurrent real-world activity --- two agents producing proposals at
the same wall-clock moment, or two users clicking a button
simultaneously --- is modelled by selecting an arbitration order
from the oracle in $\Coord^{\AR}$'s Step~1. A theorem that ranges
over all oracle traces therefore ranges over all possible
arbitration orders, which is to say over all possible concurrent
schedules. This is the architecture's answer to ``what happens if
two agents act at once?'': the meta-agent serialises them in
oracle-chosen order, and the theorems hold for every order.
\end{remark}

\begin{definition}[Reachable domain state]
\label{def:reachable}
A domain state $s \in \State_D$ is \emph{reachable} in
$\mathsf{Exec}(e_0, s_0, \omega)$ iff there exists an index~$i$
such that $s_i.\mathsf{domain} = s$.
\end{definition}

\subsubsection{Operational Semantics}
\label{sec:opsem}
\label{sec:reduction-rules}

We now give the reduction rules declared in
Definition~\ref{def:reduction}. The relation operates on
configurations $\config{e}{s}{\omega}$ (Definition~\ref{def:configuration})
and is deterministic: each non-terminal configuration has exactly
one successor.

\paragraph{Core term language.}
\label{sec:opsem-term-language}

The reduction rules refer to a small core term language sufficient to
house the architecture primitives. Here we fix only the constructors
needed by the rules below.

\begin{align*}
e \;::=\;\;
  & v
  \;\mid\; x
  \;\mid\; \mathsf{let}\; x = e_1 \;\mathsf{in}\; e_2 \\
  & \;\mid\; \mathsf{coordinate} \\
  & \;\mid\; \mathsf{adjudicate}(e_1, e_2) \\
  & \;\mid\; \mathsf{entry}(e) \\
  & \;\mid\; \mathsf{done} \\
\intertext{with values ranged over by~$v$:}
v \;::=\;\;
  & ()
  \;\mid\; (r, a)
  \;\mid\; \mathsf{approved}
  \;\mid\; \mathsf{rejected}
  \;\mid\; \mathsf{escalated}
\end{align*}

The pair value~$(r, a)$ is a coordinated proposal: a role paired
with the action that role's agent returned. The three outcome tags
$\mathsf{approved}$, $\mathsf{rejected}$, $\mathsf{escalated}$ are
values; they are the result tags attached to audit entries by
Definition~\ref{def:audit-entry}.

The term $\mathsf{coordinate}$ has no explicit arguments because, by
Definition~\ref{def:coord-ar}, $\Coord^{\AR}$ reads the current
framework state and oracle trace directly from the configuration.
The term $\mathsf{adjudicate}(e_1, e_2)$ takes a coordinated
proposal~$e_1$ (reducing to a pair $(r, a)$) and a timestamp~$e_2$
(reducing to a value in $\Time$). The term $\mathsf{entry}(e)$
indicates a commit-cycle entry point --- the client program's only
way to request that the kernel perform one coordinate/adjudicate
cycle --- and reduces to the outcome tag of the cycle. The term
$\mathsf{done}$ marks a terminal program.

This is the minimum surface needed to state the reduction rules. A
realistic client program will wrap these primitives in loops or
recursive calls; such wrappers are themselves \emph{client} code and
reduce by the standard rules of the ambient typed lambda
calculus~\cite{pierce2002types}.

\paragraph{Evaluation contexts.}
\label{sec:opsem-eval-contexts}

Reduction is defined by a small set of primitive rules together with
a compatible-closure rule under evaluation contexts:
\[
  E \;::=\; [\,\cdot\,]
       \;\mid\; \mathsf{let}\; x = E \;\mathsf{in}\; e
       \;\mid\; \mathsf{adjudicate}(E, e)
       \;\mid\; \mathsf{adjudicate}(v, E)
       \;\mid\; \mathsf{entry}(E).
\]

\begin{mathpar}
\inferrule*[left=\textsc{E-Ctx}]
  { \config{e}{s}{\omega}
      \;\longrightarrow\;
    \config{e'}{s'}{\omega'}
  }
  { \config{E[e]}{s}{\omega}
      \;\longrightarrow\;
    \config{E[e']}{s'}{\omega'}
  }
\end{mathpar}

\noindent All other rules operate on terms at the top of the
evaluation context; \textsc{E-Ctx} lifts them to arbitrary program
positions. Evaluation is left-to-right and call-by-value.

\paragraph{Administrative reductions.}
\label{sec:opsem-admin}

The following rules are standard and consume no oracle. They govern
sequencing of the client program.

\begin{mathpar}
\inferrule*[left=\textsc{E-Let}]
  {\;}
  { \config{\mathsf{let}\; x = v \;\mathsf{in}\; e}{s}{\omega}
      \;\longrightarrow\;
    \config{e[v/x]}{s}{\omega}
  }
\end{mathpar}

No other administrative rules are needed: the remaining client-level
constructors (application, case analysis, etc.) step by the standard
call-by-value rules of the typed lambda
calculus~\cite{pierce2002types}, lifted through \textsc{E-Ctx} onto
configurations without consulting~$s$ or~$\omega$.

\paragraph{Architecture reductions.}
\label{sec:opsem-arch}

The two rules below are the operational heart of \AR. They invoke,
respectively, the Agentic Redux coordination protocol
(Definition~\ref{def:coord-ar}) and the Agentic Redux commit
discipline (Definition~\ref{def:commit-ar}). Writing $\omega
\Rightarrow_{k} \omega'$ for ``$\omega'$ is the suffix of $\omega$
remaining after $k$ consultations,'' we let each rule consume
whatever oracle prefix the invoked function consumes. The consumed
count is not fixed by the calculus --- $\Coord^{\AR}$ consumes
oracle for the arbitration choice and for any oracle reads
performed by the invoked agent, while $\Commit^{\AR}$ consumes no
oracle, since invariants are pure state predicates
(\cref{rem:state-predicates}) --- but is a deterministic function
of the inputs passed to $\Coord^{\AR}$ or $\Commit^{\AR}$.

\paragraph{Coordinate.} The $\mathsf{coordinate}$ term reads the
current framework state and oracle, invokes $\Coord^{\AR}$ to
arbitrate a role, project its slice, and invoke its agent, and
reduces to the returned $(r, a)$ pair. The framework state is
unchanged: $\Coord^{\AR}$ is a pure function of its inputs and
produces a proposal, not a commit.

\begin{mathpar}
\inferrule*[left=\textsc{E-Coordinate}]
  { \Coord^{\AR}(s, \AgentPop, \omega) \;=\; (r, a)
  \\\\
    \omega \;\Rightarrow_{k}\; \omega'
      \text{ where $k$ is the number of oracle reads performed by $\Coord^{\AR}$}
  }
  { \config{\mathsf{coordinate}}{s}{\omega}
      \;\longrightarrow\;
    \config{(r, a)}{s}{\omega'}
  }
\end{mathpar}

\noindent The agent population $\AgentPop$ is fixed by the client
program's instantiation of $\AgRedux(D, \AgentPop)$
(Definition~\ref{def:ar-architecture}); it is a parameter of the
reduction relation, not a runtime value. If $\Coord^{\AR}(s,
\AgentPop, \omega)$ is undefined --- for example, because
$\AgentPop$ is empty, or because the selected agent's partial
function does not terminate on the given slice --- the configuration
is stuck; \textsc{E-Coordinate} does not apply. Stuckness of
well-typed configurations is ruled out by Type Safety
(Theorem~5).

\paragraph{Adjudicate.} The $\mathsf{adjudicate}$ term takes a
coordinated proposal~$(r, a)$ and a timestamp~$t$ and invokes
$\Commit^{\AR}$ to produce a new framework state and an outcome tag.
By Definition~\ref{def:commit-ar}, $\Commit^{\AR}$ either appends an
approval audit entry and applies the mutation, or appends a
rejection/escalation audit entry and leaves $s.\mathsf{domain}$
unchanged. We split the rule by outcome so that the three cases are
individually inspectable in proofs.

\begin{mathpar}
\inferrule*[left=\textsc{E-Adjudicate-Approve}]
  { \Commit^{\AR}(s, a, t, \omega) \;=\; s'
  \\\\
    s'.\mathsf{audit} \;=\; s.\mathsf{audit} \,{+\!+}\, [\,\text{approval entry}\,]
  \\
    s'.\mathsf{queue} \;=\; s.\mathsf{queue}
  \\
    s'.\mathsf{domain} \;=\; \applyMutation_D(s.\mathsf{domain}, a, t)
  \\\\
    \omega \;\Rightarrow_{k}\; \omega'
  }
  { \config{\mathsf{adjudicate}((r, a), t)}{s}{\omega}
      \;\longrightarrow\;
    \config{\mathsf{approved}}{s'}{\omega'}
  }
\end{mathpar}

\begin{mathpar}
\inferrule*[left=\textsc{E-Adjudicate-Reject}]
  { \Commit^{\AR}(s, a, t, \omega) \;=\; s'
  \\\\
    s'.\mathsf{audit} \;=\; s.\mathsf{audit} \,{+\!+}\, [\,\text{rejection entry}\,]
  \\
    s'.\mathsf{queue} \;=\; s.\mathsf{queue}
  \\
    s'.\mathsf{domain} \;=\; s.\mathsf{domain}
  \\\\
    \omega \;\Rightarrow_{k}\; \omega'
  }
  { \config{\mathsf{adjudicate}((r, a), t)}{s}{\omega}
      \;\longrightarrow\;
    \config{\mathsf{rejected}}{s'}{\omega'}
  }
\end{mathpar}

\begin{mathpar}
\inferrule*[left=\textsc{E-Adjudicate-Escalate}]
  { \Commit^{\AR}(s, a, t, \omega) \;=\; s'
  \\\\
    s'.\mathsf{audit} \;=\; s.\mathsf{audit} \,{+\!+}\, [\,\text{escalation entry}\,]
  \\
    s'.\mathsf{queue} \;=\; s.\mathsf{queue} \,{+\!+}\, [\,\text{pending proposal}\,]
  \\
    s'.\mathsf{domain} \;=\; s.\mathsf{domain}
  \\\\
    \omega \;\Rightarrow_{k}\; \omega'
  }
  { \config{\mathsf{adjudicate}((r, a), t)}{s}{\omega}
      \;\longrightarrow\;
    \config{\mathsf{escalated}}{s'}{\omega'}
  }
\end{mathpar}

\noindent The three rules are exhaustive and mutually exclusive:
$\Commit^{\AR}(s, a, t, \omega)$ is a deterministic function
(Definition~\ref{def:commit-ar}) whose output falls into exactly one
of the three branches. The ``approval entry,'' ``rejection entry,''
``escalation entry,'' and ``pending proposal'' appearing in the
rules are the concrete audit and queue items constructed by
$\Commit^{\AR}$; their exact shape is fixed by
Definitions~\ref{def:audit-entry} and~\ref{def:review-queue}.

\paragraph{Entry and termination.} The $\mathsf{entry}$ constructor
provides a syntactic marker for a complete commit cycle. Its rule
is administrative:

\begin{mathpar}
\inferrule*[left=\textsc{E-Entry}]
  {\;}
  { \config{\mathsf{entry}(v)}{s}{\omega}
      \;\longrightarrow\;
    \config{v}{s}{\omega}
  }
\end{mathpar}

\noindent A configuration $\config{\mathsf{done}}{s}{\omega}$ is
terminal: no rule applies. Values other than $\mathsf{done}$ are
terminal only when no enclosing evaluation context can step them via
\textsc{E-Ctx} or \textsc{E-Let}.

\paragraph{Structural lemmas.}
\label{sec:opsem-lemmas}

Five lemmas about the reduction relation will be cited by the
metatheory of \cref{sec:metatheory}. All are proved by inspection of
the rules of \cref{sec:opsem-arch}.

\begin{lemma}[Determinism]
\label{lem:determinism}
$\longrightarrow$ is a partial function on configurations.
\end{lemma}
\begin{proof}
The primitive reduction rules operate on disjoint top-level term
constructors ($\mathsf{let}$, $\mathsf{coordinate}$,
$\mathsf{adjudicate}$, $\mathsf{entry}$), so at most one primitive
rule applies at any evaluation-context redex. The three
\textsc{E-Adjudicate-*} rules are discriminated by the output of
$\Commit^{\AR}$, which is a computable partial function of its
inputs (Definition~\ref{def:commit-ar}). \textsc{E-Ctx} lifts
primitive rules through evaluation contexts left-to-right and
call-by-value; the evaluation context for a given non-value term is
unique. Each rule that consults the oracle does so through a single
call to $\Coord^{\AR}$ or $\Commit^{\AR}$, each a computable partial
function of its inputs (Definitions~\ref{def:coord},
\ref{def:commit}); the number of oracle elements consumed, and the
resulting suffix~$\omega'$, are therefore functions of the
configuration alone, not additional nondeterministic choices.
\end{proof}

\begin{remark}[Rule classification]
\label{rem:rule-classification}
The reduction rules of \cref{sec:opsem-arch} fall into two
classes by the shape of their conclusion's domain clause: a rule
is \emph{applyMutation-applying} if its conclusion fixes
$s'.\mathsf{domain} = \applyMutation_D(s.\mathsf{domain}, a, t)$
for some action and timestamp, and \emph{domain-preserving}
otherwise --- meaning its conclusion either fixes
$s'.\mathsf{domain} = s.\mathsf{domain}$ explicitly or leaves
$s$ entirely unchanged. The unique applyMutation-applying rule
of the base calculus is \textsc{E-Adjudicate-Approve};
\textsc{E-Let}, \textsc{E-Coordinate},
\textsc{E-Adjudicate-Reject},
\textsc{E-Adjudicate-Escalate}, and \textsc{E-Entry} are
domain-preserving. \textsc{E-Ctx} inherits the class of the
rule it lifts.
\end{remark}

The metatheory below depends on the rule set only through this
classification: an extension that declares each new rule
applyMutation-applying or domain-preserving extends every
result whose proof appeals only to the classification, with
Lemma~\ref{lem:domain-localization} and the domain-preserving
case of Theorem~1 extending verbatim. The applyMutation-applying
case requires the extension to discharge its own analogue of
Lemma~\ref{lem:approval-soundness}, a non-classification
property.

\begin{lemma}[Domain-state localization]
\label{lem:domain-localization}
For every reduction step
$\config{e}{s}{\omega} \longrightarrow \config{e'}{s'}{\omega'}$
generated by a domain-preserving rule, possibly lifted through
\textsc{E-Ctx}, $s'.\mathsf{domain} = s.\mathsf{domain}$.
\end{lemma}
\begin{proof}
By case analysis on the primitive rule invoked, with
\textsc{E-Ctx} as the inductive step. Rules \textsc{E-Let} and
\textsc{E-Entry} leave $s$ entirely unchanged. Rule
\textsc{E-Coordinate} carries $s$ through its conclusion without
modification. Rules \textsc{E-Adjudicate-Reject} and
\textsc{E-Adjudicate-Escalate} both include
$s'.\mathsf{domain} = s.\mathsf{domain}$ as an explicit premise.
\textsc{E-Ctx} preserves
$s'.\mathsf{domain} = s.\mathsf{domain}$ by the inductive
hypothesis applied to the inner step. This is the formal
content, at the level of the reduction relation, of ``the
meta-agent is the only component permitted to produce an
$\approved$-tagged outcome'' (Definition~\ref{def:meta-agent}):
the applyMutation-applying rule \textsc{E-Adjudicate-Approve} is
the sole route by which $s.\mathsf{domain}$ can change in the
base calculus.
\end{proof}

\begin{lemma}[Approval soundness]
\label{lem:approval-soundness}
If
$\config{\mathsf{adjudicate}((r, a), t)}{s}{\omega} \longrightarrow
\config{\mathsf{approved}}{s'}{\omega'}$
under \textsc{E-Adjudicate-Approve}, then
$s'.\mathsf{domain} \models \invariants_D$.
\end{lemma}
\begin{proof}
The \textsc{E-Adjudicate-Approve} rule applies when
$\Commit^{\AR}(s, a, t, \omega)$ follows the $\pass$-all branch of
Definition~\ref{def:commit-ar}. By Step~1 of that definition,
$\Commit^{\AR}$ computes the candidate
$\hat{s} = \applyMutation_D(s.\mathsf{domain}, a, t)$ before walking
the invariant list. By Step~2, each invariant is evaluated
on~$\hat{s}$. The $\pass$-all branch (Step~4) is reached precisely
when every $\iota \in \invariants_D$ returns $\pass$ on~$\hat{s}$,
which by Definition~\ref{def:invariant-satisfaction} is
$\hat{s} \models \invariants_D$. The rule's premises fix
$s'.\mathsf{domain} = \applyMutation_D(s.\mathsf{domain}, a, t) =
\hat{s}$, yielding $s'.\mathsf{domain} \models \invariants_D$.
\end{proof}

\begin{lemma}[Non-approval witness]
\label{lem:non-approval-witness}
If
$\config{\mathsf{adjudicate}((r, a), t)}{s}{\omega} \longrightarrow
\config{\mathsf{rejected}}{s'}{\omega'}$
under \textsc{E-Adjudicate-Reject}, or
$\config{\mathsf{adjudicate}((r, a), t)}{s}{\omega} \longrightarrow
\config{\mathsf{escalated}}{s'}{\omega'}$
under \textsc{E-Adjudicate-Escalate}, then there exist
$\iota_k \in \invariants_D$ and $m \in \mathsf{String}$ such that,
letting $\hat{s} = \applyMutation_D(s.\mathsf{domain}, a, t)$:
\begin{itemize}
  \item on the \textsc{E-Adjudicate-Reject} branch,
        $\iota_k(\hat{s}) = \rejectRes(m)$;
  \item on the \textsc{E-Adjudicate-Escalate} branch,
        $\iota_k(\hat{s}) = \escalateRes(m)$;
  \item in particular $\iota_k(\hat{s}) \ne \pass$; and
  \item the entry appended to $s'.\mathsf{audit}$ by this step has
        witness slot $\mathsf{Some}(\mathrm{id}(\iota_k), m)$.
\end{itemize}
\end{lemma}
\begin{proof}
In either case, $\Commit^{\AR}(s, a, t, \omega)$ follows the
non-$\pass$-all branch of Definition~\ref{def:commit-ar}. Step~1
computes $\hat{s} = \applyMutation_D(s.\mathsf{domain}, a, t)$;
Step~2 evaluates each invariant on~$\hat{s}$; Step~3 halts on the
first invariant $\iota_k \in \invariants_D$ with
$\iota_k(\hat{s}) \ne \pass$. By Definition~\ref{def:inv-result},
$\iota_k(\hat{s})$ must have the form $\rejectRes(m)$ or
$\escalateRes(m)$ for some $m \in \mathsf{String}$. On the reject
branch of Step~3, Definition~\ref{def:commit-ar} appends to
$s.\mathsf{audit}$ an entry with witness slot
$\mathsf{Some}(\mathrm{id}(\iota_k), m)$; analogously on the
escalate branch. In both cases $\iota_k(\hat{s}) \ne \pass$.
\end{proof}

\begin{lemma}[Audit append discipline]
\label{lem:audit-append}
For every reduction step
$\config{e}{s}{\omega} \longrightarrow \config{e'}{s'}{\omega'}$:
\begin{itemize}
  \item if the step is an instance of
        \textsc{E-Adjudicate-Approve},
        \textsc{E-Adjudicate-Reject}, or
        \textsc{E-Adjudicate-Escalate} (possibly lifted through
        \textsc{E-Ctx}), then
        $s'.\mathsf{audit} = s.\mathsf{audit} \,{+\!+}\,
        [\,\text{entry}\,]$ for exactly one entry, fixed by
        Definition~\ref{def:commit-ar};
  \item otherwise, $s'.\mathsf{audit} = s.\mathsf{audit}$.
\end{itemize}
\end{lemma}
\begin{proof}
By case analysis on the primitive rule invoked, with
\textsc{E-Ctx} as the inductive step. The three
\textsc{E-Adjudicate-*} rules each have premise
$s'.\mathsf{audit} = s.\mathsf{audit} \,{+\!+}\,
[\,\text{entry}\,]$ for one entry whose shape is fixed by
Definition~\ref{def:commit-ar}. The remaining rules
(\textsc{E-Let}, \textsc{E-Coordinate}, \textsc{E-Entry}) carry
$s$ through their conclusions syntactically, so the audit field
is unchanged. \textsc{E-Ctx} preserves both cases by the
inductive hypothesis applied to the inner step.
\end{proof}

\begin{remark}[Slices are computed at invocation time]
\label{rem:slice-freshness}
The $\mathsf{coordinate}$ rule reads~$s$ from the enclosing
configuration at the moment the rule fires, and $\Coord^{\AR}$
computes the slice $\project_D(s.\mathsf{domain}, r)$ at that same
moment. No rule captures a slice into a term for later use: the
term language has no value form for slices. An agent can
therefore only be invoked on a slice derived from the current
framework state.
\end{remark}

\subsubsection{Typing Rules}
\label{sec:typing}

The typing judgment for the calculus takes the form
$\Gamma \vdash_D e : \tau$, parameterised by a fixed $D : \Domain$
and a variable context $\Gamma$ mapping variables to types. Types
are drawn from the following grammar:
\[
  \tau \;::=\; \mathsf{Unit} \;\mid\; \Roles_D
         \;\mid\; \Action_D \;\mid\; \Time
         \;\mid\; \tau_1 \times \tau_2 \;\mid\; \Outcome_D.
\]
The types $\Roles_D$, $\Action_D$, and $\Time$ are base types whose
inhabitants are the role, action, and timestamp constants supplied
by the domain and the ambient calculus. The type $\mathsf{Unit}$ has
inhabitants $()$ and $\mathsf{done}$.

\begin{remark}[Outcome type at the calculus level]
\label{rem:outcome-calculus}
The type $\Outcome_D$ appearing in the typing grammar above is the
three-tag type inhabited by the outcome constants
$\approved, \rejected, \escalated$. This is the calculus-level
projection of the meta-agent's output type
(Definition~\ref{def:meta-agent}), whose branches additionally
carry payloads. At reduction time, the payload information is
routed into the framework-state update performed by
$\Commit^{\AR}$ (Definition~\ref{def:commit-ar}), while the three
tags alone remain as the term-level residual.
\end{remark}

The rules for variables, let-bindings, and the two inhabitants of
$\mathsf{Unit}$ are standard.

\begin{mathpar}
\inferrule*[left=\textsc{T-Var}]
  { x{:}\tau \in \Gamma }
  { \Gamma \vdash_D x : \tau }

\inferrule*[left=\textsc{T-Unit}]
  {\ }
  { \Gamma \vdash_D () : \mathsf{Unit} }

\inferrule*[left=\textsc{T-Done}]
  {\ }
  { \Gamma \vdash_D \mathsf{done} : \mathsf{Unit} }

\inferrule*[left=\textsc{T-Let}]
  { \Gamma \vdash_D e_1 : \tau_1
    \\
    \Gamma, x{:}\tau_1 \vdash_D e_2 : \tau_2
  }
  { \Gamma \vdash_D \mathsf{let}\; x = e_1 \;\mathsf{in}\; e_2 : \tau_2 }
\end{mathpar}

The base-type constants are typed by axiom.

\begin{mathpar}
\inferrule*[left=\textsc{T-Role}]
  { r \in \Roles_D }
  { \Gamma \vdash_D r : \Roles_D }

\inferrule*[left=\textsc{T-Action}]
  { a \in \Action_D }
  { \Gamma \vdash_D a : \Action_D }

\inferrule*[left=\textsc{T-Time}]
  { t \in \Time }
  { \Gamma \vdash_D t : \Time }

\inferrule*[left=\textsc{T-Approved}]
  {\ }
  { \Gamma \vdash_D \approved : \Outcome_D }

\inferrule*[left=\textsc{T-Rejected}]
  {\ }
  { \Gamma \vdash_D \rejected : \Outcome_D }

\inferrule*[left=\textsc{T-Escalated}]
  {\ }
  { \Gamma \vdash_D \escalated : \Outcome_D }
\end{mathpar}

A coordinated proposal $(r, a)$ is typed by the product-introduction
rule specialised to the role/action pair.

\begin{mathpar}
\inferrule*[left=\textsc{T-Proposal}]
  { \Gamma \vdash_D r : \Roles_D
    \\
    \Gamma \vdash_D a : \Action_D
  }
  { \Gamma \vdash_D (r, a) : \Roles_D \times \Action_D }
\end{mathpar}

The three architecture primitives are typed as follows.

\begin{mathpar}
\inferrule*[left=\textsc{T-Coordinate}]
  {\ }
  { \Gamma \vdash_D \mathsf{coordinate} : \Roles_D \times \Action_D }

\inferrule*[left=\textsc{T-Adjudicate}]
  { \Gamma \vdash_D e_1 : \Roles_D \times \Action_D
    \\
    \Gamma \vdash_D e_2 : \Time
  }
  { \Gamma \vdash_D \mathsf{adjudicate}(e_1, e_2) : \Outcome_D }

\inferrule*[left=\textsc{T-Entry}]
  { \Gamma \vdash_D e : \Outcome_D }
  { \Gamma \vdash_D \mathsf{entry}(e) : \Outcome_D }
\end{mathpar}

\textsc{T-Coordinate} assigns the coordinate primitive the type of
its reduct, a coordinated proposal. \textsc{T-Adjudicate} requires
a coordinated proposal and a timestamp and returns an outcome tag.
\textsc{T-Entry} preserves the outcome type of its argument,
marking a complete commit cycle.

\begin{remark}[What the typing discipline does and does not enforce]
\label{rem:typing-scope}
The typing rules rule out ill-formed combinations such as
$\mathsf{adjudicate}(\approved, \approved)$, which would be
operationally stuck under the rules of \cref{sec:opsem-arch}. They
do not enforce domain-level well-formedness of the action in a
proposal $(r, a)$; in particular, they do not require that
$a \in \Action_D(r)$ in the sense of
Definition~\ref{def:agent}. That refinement is enforced
operationally by $\Coord^{\AR}$, which produces proposals by
invoking $\alpha_r$ on $\Slice_D(r)$ and whose returned action is
classified as belonging to role~$r$ by Condition~1 of
Definition~\ref{def:agent}. Pushing the role-indexing into the
type system would require dependent products and is not necessary
for any result in this paper.
\end{remark}

\subsubsection{Vocabulary for Theorem Statements}
\label{sec:theorem-vocabulary}

Every concept appearing in the informal theorem statements of
\cref{sec:metatheory} has a formal counterpart above. The following
table records the dependency of each theorem on the preceding
definitions.

\begin{center}
\small
\begin{tabular}{lll}
\toprule
Result & Scope & Required definitions \\
\midrule
T1 Invariant Preservation  & \Adjudicable
  & \ref{def:adjudicable}, \ref{def:inv-result},
    \ref{def:invariant-satisfaction},
    \ref{def:commit-ar}, \ref{def:reduction},
    \ref{def:execution}, \ref{def:reachable} \\
P1 Agent Confinement    & \Domain
  & \ref{def:footprint}, \ref{def:slice},
    \ref{def:wf-roles}, \ref{def:domain-signature},
    \ref{def:agent}, \ref{def:oracle},
    \ref{def:reduction} \\
T2 Audit Log Integrity  & \Adjudicable
  & \ref{def:adjudicable}, \ref{def:inv-result},
    \ref{def:invariant-satisfaction},
    \ref{def:audit-entry}, \ref{def:commit-ar},
    \ref{def:reduction}, \ref{def:execution} \\
T3 Type Safety          & \Domain
  & \ref{def:configuration}, \ref{def:wt-config},
    \ref{def:reduction}, \ref{sec:typing} \\
\midrule
C1 Write Skew Freedom   & \Domain
  & \ref{def:local-crosscutting},
    \ref{def:write-skew-calc},
    \ref{def:execution}, \ref{def:reachable}, T1 \\
\bottomrule
\end{tabular}
\end{center}
\subsection{Metatheory}
\label{sec:metatheory}

\subsubsection{Invariant Preservation}
\label{app:proof-invariant-safety}
\paragraph{Theorem 1 (Invariant Preservation).}
\emph{For every $D : \Adjudicable$, every client program~$e_0$
well-typed over $\AgRedux(D)$, every initial framework state~$s_0$
of type $\FrameworkState_D$ with
$s_0.\mathsf{domain} \models \invariants_D$, and every oracle trace
$\omega \in \Oracle$: every reachable domain state
in~$\mathsf{Exec}(e_0, s_0, \omega)$ satisfies $\invariants_D$.}

\begin{proof}
Let $s$ be a reachable domain state
in~$\mathsf{Exec}(e_0, s_0, \omega)$. By
Definition~\ref{def:reachable}, there exists an index $i \geq 0$
with $s = s_i.\mathsf{domain}$. We show, by induction on~$i$, that
$s_j.\mathsf{domain} \models \invariants_D$ for every
$0 \leq j \leq i$; specialising to $j = i$ yields
$s \models \invariants_D$.

\emph{Base case} ($j = 0$). By hypothesis,
$s_0.\mathsf{domain} \models \invariants_D$.

\emph{Inductive step.} Assume
$s_j.\mathsf{domain} \models \invariants_D$ for some $j < i$.
The execution's step at index~$j$ is a reduction
$\config{e_j}{s_j}{\omega_j} \longrightarrow
\config{e_{j+1}}{s_{j+1}}{\omega_{j+1}}$, and is classified by
whether it is an instance of \textsc{E-Adjudicate-Approve}.

\textbf{Case 1: the step is generated by a domain-preserving
rule} (Remark~\ref{rem:rule-classification}).
By Lemma~\ref{lem:domain-localization},
$s_{j+1}.\mathsf{domain} = s_j.\mathsf{domain}$. The inductive
hypothesis gives
$s_{j+1}.\mathsf{domain} \models \invariants_D$.

\textbf{Case 2: the step is generated by an
applyMutation-applying rule.} In the base calculus the unique
such rule is \textsc{E-Adjudicate-Approve}, possibly lifted
through \textsc{E-Ctx}. The primitive-rule instance then has the
form
\[
  \config{\mathsf{adjudicate}((r, a), t)}{s_j}{\omega_j}
  \;\longrightarrow\;
  \config{\mathsf{approved}}{s_{j+1}}{\omega_{j+1}}
\]
for some role~$r$, action~$a$, and timestamp~$t$; \textsc{E-Ctx}
(\cref{sec:opsem-eval-contexts}) passes the framework state and
oracle through unchanged. By
Lemma~\ref{lem:approval-soundness},
$s_{j+1}.\mathsf{domain} \models \invariants_D$.

In either case,
$s_{j+1}.\mathsf{domain} \models \invariants_D$, completing the
induction.
\end{proof}

\begin{remark}[Reachable-state form]
\label{rem:t1-reachable-form}
Theorem~1 is stated about reachable domain states rather than about
a relation between each approval step's pre-state and its committed
action. The reachable-state form is what a compliance argument
actually needs to cite (``no reachable state violates an
invariant''), and it is available only because
Lemma~\ref{lem:approval-soundness} establishes invariant
satisfaction on the \emph{post-state} of each approval, rather than
on the pre-state-plus-action pair that was evaluated. The
soundness of the commit-discipline post-check
(Remark~\ref{rem:state-predicates}) is therefore what drives the
inductive step; a per-approval-step statement phrased on the
pre-state-plus-action pair would reveal only a weaker, contingent
relationship.
\end{remark}

\subsubsection{Agent Confinement}
\label{app:proof-confinement}
\paragraph{Proposition 1 (Agent Confinement).}
\emph{For every $D : \Domain$, the following two confinement
properties hold.}

\begin{enumerate}
\item[(i)] \emph{\textbf{Observation confinement.} For every role
$r \in \Roles_D$, every agent $\alpha_r \in \Agent_D(r)$, every
oracle trace $\omega \in \Oracle$, and every pair of domain states
$s_1, s_2 \in \State_D$ with
$\project_D(s_1, r) = \project_D(s_2, r)$: if
$\alpha_r(\project_D(s_1, r), \omega)$ is defined then
$\alpha_r(\project_D(s_2, r), \omega)$ is defined and equal.}

\item[(ii)] \emph{\textbf{Write-authority confinement.} For every
primitive reduction
$\config{\mathsf{coordinate}}{s}{\omega} \longrightarrow
\config{(r, a)}{s'}{\omega'}$
under \textsc{E-Coordinate}: $s' = s$. In particular,
$s'.\mathsf{domain} = s.\mathsf{domain}$.}
\end{enumerate}

\begin{proof}
\emph{Part~(i).} Immediate from Definition~\ref{def:agent}: an
agent $\alpha_r$ is a function whose first argument is its
role's slice. By hypothesis, the two slices agree, so the
applications $\alpha_r(\project_D(s_1, r), \omega)$ and
$\alpha_r(\project_D(s_2, r), \omega)$ are at the same argument
pair, hence equal where defined.

\emph{Part~(ii).} Immediate from \textsc{E-Coordinate}
(\cref{sec:opsem-arch}): the conclusion of the rule is
$\config{(r, a)}{s}{\omega'}$, threading the framework state~$s$
unchanged from the premise. The equality $s' = s$ is therefore
syntactic, and $s'.\mathsf{domain} = s.\mathsf{domain}$ follows.
This case of Lemma~\ref{lem:domain-localization} is thus
strengthened: invoking an agent mutates neither the domain state
nor any other component of the framework state.
\end{proof}

\subsubsection{Write Skew Freedom}
\label{app:proof-write-skew}

A write-skew execution at the calculus level is the formal
counterpart of Berenson \emph{et al.}'s anomaly~A5B in the
database-isolation setting~\cite{berenson}: a committed state
violating a constraint that references state modified by more
than one concurrent transaction --- a constraint no
transaction, operating on its own view, could have caused to
fail alone. In the calculus, a sub-agent observes exactly its
role's slice (Definitions~\ref{def:slice},~\ref{def:agent}),
so the calculus-level analogue of such a constraint is a
cross-cutting invariant
(Definition~\ref{def:local-crosscutting}). Scheduling
non-determinism that in the classical setting arises from
concurrent transactions is absorbed here into the oracle, which
selects the arbitration order in which $\Coord^{\AR}$ elicits
proposals (Remark~\ref{rem:seq-and-simultaneity}); ranging over
all $\omega \in \Oracle$ therefore ranges over all possible
schedules.

\begin{definition}[Write-skew execution at the calculus level]
\label{def:write-skew-calc}
An execution $\mathsf{Exec}(e_0, s_0, \omega)$ of $\AgRedux(D)$
is a \emph{write-skew execution} iff there exist a reachable
domain state~$s$ and a cross-cutting invariant
$\iota \in \invariants_D$ with $\iota(s) \ne \pass$.
\end{definition}

\paragraph{Corollary 1 (Write Skew Freedom).}
\emph{For every $D : \Domain$, every client program~$e_0$
well-typed over $\AgRedux(D)$, every initial framework state~$s_0$
of type $\FrameworkState_D$ with $s_0.\mathsf{domain} \models
\invariants_D$, and every oracle trace $\omega \in \Oracle$:
$\mathsf{Exec}(e_0, s_0, \omega)$ is not a write-skew execution.}

\begin{proof}
By Theorem~1, every reachable domain state in
$\mathsf{Exec}(e_0, s_0, \omega)$ satisfies $\invariants_D$,
hence satisfies every cross-cutting invariant in particular.
The execution therefore does not meet
Definition~\ref{def:write-skew-calc}.
\end{proof}

\subsubsection{Audit Log Integrity}
\label{app:proof-audit-integrity}
\paragraph{Theorem 2 (Audit Log Integrity).}
\emph{For every $D : \Adjudicable$, every client program~$e_0$
well-typed over $\AgRedux(D)$, every initial framework state~$s_0$
of type $\FrameworkState_D$, and every oracle trace $\omega \in
\Oracle$, the execution $\mathsf{Exec}(e_0, s_0, \omega)$
satisfies:}

\emph{\textbf{(i) Faithful approval entries.} For every step
$\config{e_j}{s_j}{\omega_j} \longrightarrow
\config{e_{j+1}}{s_{j+1}}{\omega_{j+1}}$
in the execution that is an instance of
\textsc{E-Adjudicate-Approve} (possibly lifted through
\textsc{E-Ctx}) with coordinated proposal $(r, a)$ and
timestamp~$t$: the sole entry appended to
$s_{j+1}.\mathsf{audit}$ is tagged $\approved$, carries
action~$a$ and timestamp~$t$, and has witness slot
$\mathsf{None}$.}

\emph{\textbf{(ii) Faithful non-approval entries with witness.}
For every step
$\config{e_j}{s_j}{\omega_j} \longrightarrow
\config{e_{j+1}}{s_{j+1}}{\omega_{j+1}}$
in the execution that is an instance of
\textsc{E-Adjudicate-Reject} (resp.\
\textsc{E-Adjudicate-Escalate}), possibly lifted through
\textsc{E-Ctx}, with coordinated proposal $(r, a)$ and
timestamp~$t$: the sole entry appended to
$s_{j+1}.\mathsf{audit}$ is tagged $\rejected$ (resp.\
$\escalated$), carries action~$a$ and timestamp~$t$, and has
witness slot $\mathsf{Some}(\mathrm{id}(\iota_k), m)$ where
$\iota_k \in \invariants_D$ satisfies
$\iota_k(\applyMutation_D(s_j.\mathsf{domain}, a, t)) =
\rejectRes(m)$ (resp.\ $\escalateRes(m)$); in particular
$\iota_k(\applyMutation_D(s_j.\mathsf{domain}, a, t)) \ne \pass$.}

\emph{\textbf{(iii) Exclusivity of provenance.} For every
reduction step
$\config{e_j}{s_j}{\omega_j} \longrightarrow
\config{e_{j+1}}{s_{j+1}}{\omega_{j+1}}$
in the execution that is not an instance of
\textsc{E-Adjudicate-Approve}, \textsc{E-Adjudicate-Reject}, or
\textsc{E-Adjudicate-Escalate} (possibly lifted through
\textsc{E-Ctx}): $s_{j+1}.\mathsf{audit} = s_j.\mathsf{audit}$.}

\emph{\textbf{(iv) Append-only monotonicity.} For every reduction
step in the execution, $s_j.\mathsf{audit}$ is a prefix of
$s_{j+1}.\mathsf{audit}$.}

\emph{\textbf{(v) Per-step exactness.} For every reduction step
in the execution that is an instance of
\textsc{E-Adjudicate-Approve}, \textsc{E-Adjudicate-Reject}, or
\textsc{E-Adjudicate-Escalate} (possibly lifted through
\textsc{E-Ctx}):
$|s_{j+1}.\mathsf{audit}| = |s_j.\mathsf{audit}| + 1$.}

\begin{proof}
Fix a step
$\config{e_j}{s_j}{\omega_j} \longrightarrow
\config{e_{j+1}}{s_{j+1}}{\omega_{j+1}}$
of the execution, derived from a primitive reduction rule,
possibly lifted through \textsc{E-Ctx}
(\cref{sec:opsem-eval-contexts}), which passes the framework state
and oracle through unchanged.

\emph{Clause (i).} The primitive rule is
\textsc{E-Adjudicate-Approve}. By Step~4 of
Definition~\ref{def:commit-ar}, $\Commit^{\AR}$ appends to
$s_j.\mathsf{audit}$ exactly one entry, tagged $\approved$, with
action~$a$, timestamp~$t$, and witness slot $\mathsf{None}$; this
is the entry fixed by the rule's premise
$s_{j+1}.\mathsf{audit} = s_j.\mathsf{audit} \,{+\!+}\,
[\,\text{approval entry}\,]$.

\emph{Clause (ii).} The primitive rule is
\textsc{E-Adjudicate-Reject} or \textsc{E-Adjudicate-Escalate}. By
Lemma~\ref{lem:non-approval-witness}, the halt-witnessing
invariant $\iota_k \in \invariants_D$ satisfies
$\iota_k(\hat{s}) = \rejectRes(m)$ (resp.\ $\escalateRes(m)$)
where $\hat{s} = \applyMutation_D(s_j.\mathsf{domain}, a, t)$, and
the sole entry appended to $s_{j+1}.\mathsf{audit}$ has tag
$\rejected$ (resp.\ $\escalated$), action~$a$, timestamp~$t$, and
witness slot $\mathsf{Some}(\mathrm{id}(\iota_k), m)$.

\emph{Clauses (iii)--(v).} Immediate from
Lemma~\ref{lem:audit-append}: its
\textsc{E-Adjudicate-*} branch fixes
$s_{j+1}.\mathsf{audit} = s_j.\mathsf{audit} \,{+\!+}\,
[\,\text{entry}\,]$ for exactly one entry, supplying the prefix
property of Clause~(iv) and the cardinality identity of
Clause~(v); the complementary branch fixes
$s_{j+1}.\mathsf{audit} = s_j.\mathsf{audit}$, supplying
Clause~(iii) and the trivial-prefix case of Clause~(iv).
\end{proof}

\begin{remark}[Scope of Theorem~2 and ``linearly in time'']
\label{rem:t2-approval-overlap}
Theorem~2 is purely structural: it characterises which
reduction steps append to the audit log and what those entries
look like. The companion semantic claim --- that an approval
step leaves the domain state satisfying $\invariants_D$ --- is
the content of Theorem~1 (via
Lemma~\ref{lem:approval-soundness}); it is not restated here.
Clauses~(iii)--(v) together make precise the ``linearly in
time'' claim of the informal body-text
(\cref{sec:linear-auditability}): the audit log grows by
exactly one entry per adjudicate step, from no other reduction
rule, and existing entries are never modified.
\end{remark}

\subsubsection{Type Safety}
\label{app:proof-type-safety}

The metatheory above establishes the safety and auditability
properties on which compliance arguments for \AR{} depend:
invariant preservation and its write-skew corollary (Theorem~1,
Corollary~1), audit log integrity (Theorem~2), and agent
confinement (Proposition~1). The proofs of these results refer
to~$D$ only through the operations declared by \Adjudicable{}
(Theorems~1 and~2, Corollary~1) and additionally by \Domain{}
(Proposition~1); no step case-analyzes on a particular~$D$.
Instantiating \AR{} at any~$D$ of the appropriate signature
therefore yields a system for which every corresponding result
holds. We close the metatheory with Type Safety, the standard
Progress-and-Preservation result for the calculus.

Type Safety is not cited by any other result in this paper. It
is presented here for two reasons. First, it documents that the
calculus's term grammar and typing rules interact cleanly with
the reduction relation. Second, the architecture variants
developed in~\cite{semantic-correctness-2} --- Optimistic
Parallel Redux and Hierarchical Agents --- extend the term
grammar of this paper with new constructors (respectively, a
parallel-dispatch form with merge and retry, and a delegation
form). Under those extensions the Progress cases acquire genuine
content: a merge can fail, and a delegation can fall through a
level of the hierarchy. Type Safety then becomes the theorem
under which the new term forms are shown to interact cleanly
with their extended reduction rules. The present result is the
base case of that sequence.

\paragraph{Theorem 3 (Type Safety).}
\emph{For every $D : \Domain$, every well-typed configuration
over $\AgRedux(D)$ satisfies:}
\begin{itemize}
  \item \emph{\textup{(Preservation)} If
        $\config{e}{s}{\omega}$ is well-typed over $\AgRedux(D)$
        with $\emptyset \vdash_D e : \tau$, and
        $\config{e}{s}{\omega} \longrightarrow
        \config{e'}{s'}{\omega'}$, then
        $\config{e'}{s'}{\omega'}$ is well-typed over
        $\AgRedux(D)$ with $\emptyset \vdash_D e' : \tau$.}
  \item \emph{\textup{(Progress)} If $\config{e}{s}{\omega}$ is
        well-typed over $\AgRedux(D)$ and $e$ is not a value,
        then either there exists a reduction step
        $\config{e}{s}{\omega} \longrightarrow
        \config{e'}{s'}{\omega'}$, or the configuration is
        operationally stuck on an undefined invocation of
        $\Coord^{\AR}$, $\Commit^{\AR}$, or an agent in
        $\AgentPop$.}
\end{itemize}

The proof uses three standard auxiliary lemmas.

\begin{lemma}[Canonical forms]
\label{lem:canonical-forms}
If $\emptyset \vdash_D v : \tau$ for a value~$v$, then:
\begin{itemize}
  \item if $\tau = \mathsf{Unit}$, then
        $v \in \{(), \mathsf{done}\}$;
  \item if $\tau = \Roles_D \times \Action_D$, then
        $v = (r, a)$ for some $r \in \Roles_D$ and
        $a \in \Action_D$;
  \item if $\tau = \Time$, then $v$ is a timestamp constant;
  \item if $\tau = \Outcome_D$, then
        $v \in \{\approved, \rejected, \escalated\}$.
\end{itemize}
\end{lemma}
\begin{proof}
By inspection of \cref{sec:typing}. For each type, the only
rules whose conclusions match a value of that type are the ones
listed: \textsc{T-Unit} and \textsc{T-Done} for $\mathsf{Unit}$;
\textsc{T-Proposal} for $\Roles_D \times \Action_D$;
\textsc{T-Time} for $\Time$; and
\textsc{T-Approved}, \textsc{T-Rejected}, \textsc{T-Escalated}
for $\Outcome_D$.
\end{proof}

\begin{lemma}[Substitution]
\label{lem:substitution}
If $\Gamma, x{:}\tau_1 \vdash_D e : \tau_2$ and
$\Gamma \vdash_D v : \tau_1$ for a value~$v$, then
$\Gamma \vdash_D e[v/x] : \tau_2$.
\end{lemma}
\begin{proof}
By induction on the derivation of
$\Gamma, x{:}\tau_1 \vdash_D e : \tau_2$, standard for
simply-typed calculi~\cite[Lemma~9.3.8]{pierce2002types}. The
architecture primitives \textsc{T-Coordinate},
\textsc{T-Adjudicate}, and \textsc{T-Entry} have no dependence
on~$\Gamma$ beyond the types of their subterms, so each case
reduces to the inductive hypothesis applied to those subterms.
\end{proof}

A third ingredient is a decomposition lemma for evaluation
contexts: a well-typed term filling an evaluation context has a
well-defined type for the hole, and substituting a term of the
same type into the hole preserves the overall type. This is
standard for evaluation-context formulations of
reduction~\cite{wrightfelleisen1994}.

\begin{lemma}[Evaluation-context decomposition]
\label{lem:ctx-decomp}
If $\emptyset \vdash_D E[e_1] : \tau$, then there exists a type
$\tau_1$ such that $\emptyset \vdash_D e_1 : \tau_1$, and for
every $e_1'$ with $\emptyset \vdash_D e_1' : \tau_1$,
$\emptyset \vdash_D E[e_1'] : \tau$.
\end{lemma}
\begin{proof}
By induction on~$E$ (\cref{sec:opsem-eval-contexts}).
\begin{itemize}
  \item $E = [\cdot]$: take $\tau_1 = \tau$; immediate.
  \item $E = \mathsf{let}\; x = E' \;\mathsf{in}\; e_2$: by
        inversion of \textsc{T-Let}, there exists $\tau_1'$ with
        $\emptyset \vdash_D E'[e_1] : \tau_1'$ and
        $x{:}\tau_1' \vdash_D e_2 : \tau$. The inductive
        hypothesis on~$E'$ supplies~$\tau_1$ and the
        substitution property; \textsc{T-Let} reconstructs the
        conclusion.
  \item $E = \mathsf{adjudicate}(E', e_2)$: by inversion of
        \textsc{T-Adjudicate},
        $\emptyset \vdash_D E'[e_1] : \Roles_D \times \Action_D$
        and $\emptyset \vdash_D e_2 : \Time$, with
        $\tau = \Outcome_D$. Apply the IH to~$E'$ at type
        $\Roles_D \times \Action_D$.
  \item $E = \mathsf{adjudicate}(v, E')$: by inversion of
        \textsc{T-Adjudicate},
        $\emptyset \vdash_D v : \Roles_D \times \Action_D$ and
        $\emptyset \vdash_D E'[e_1] : \Time$. Apply the IH
        to~$E'$ at type $\Time$.
  \item $E = \mathsf{entry}(E')$: by inversion of
        \textsc{T-Entry}, $\emptyset \vdash_D E'[e_1] :
        \Outcome_D$ with $\tau = \Outcome_D$. Apply the IH
        to~$E'$ at type $\Outcome_D$.
\end{itemize}
\end{proof}

\begin{proof}[Proof of Theorem~3]
We prove Preservation and Progress in turn.

\medskip
\noindent\textbf{Preservation.}
Suppose $\config{e}{s}{\omega}$ is well-typed with
$\emptyset \vdash_D e : \tau$ and $s : \FrameworkState_D$, and a
reduction
$\config{e}{s}{\omega} \longrightarrow
 \config{e'}{s'}{\omega'}$
is derived. We case-analyse on the primitive reduction rule
invoked, with \textsc{E-Ctx} as the inductive step.

\smallskip
\emph{Case \textsc{E-Let}.} $e = \mathsf{let}\; x = v
\;\mathsf{in}\; e_2$ and $e' = e_2[v/x]$. By inversion of
\textsc{T-Let}, $\emptyset \vdash_D v : \tau_1$ and
$x{:}\tau_1 \vdash_D e_2 : \tau$. By
Lemma~\ref{lem:substitution},
$\emptyset \vdash_D e_2[v/x] : \tau$. Framework state unchanged:
$s' = s$.

\smallskip
\emph{Case \textsc{E-Coordinate}.} $e = \mathsf{coordinate}$ and
$e' = (r, a)$ where $(r, a) = \Coord^{\AR}(s, \AgentPop, \omega)$.
By \textsc{T-Coordinate}, $\tau = \Roles_D \times \Action_D$.
Definition~\ref{def:coord-ar} specifies that $\Coord^{\AR}$
selects $r \in \Roles_D$ and returns the action~$a$ produced by
$\alpha_r$; by the agent's type signature and Condition~1 of
Definition~\ref{def:agent}, $a$ is an element of $\Action_D$
classified as belonging to role~$r$. Thus \textsc{T-Role} gives
$\emptyset \vdash_D r : \Roles_D$, \textsc{T-Action} gives
$\emptyset \vdash_D a : \Action_D$, and \textsc{T-Proposal}
concludes
$\emptyset \vdash_D (r, a) : \Roles_D \times \Action_D$.
Framework state unchanged.

\smallskip
\emph{Case \textsc{E-Adjudicate-Approve}.}
$e = \mathsf{adjudicate}((r, a), t)$ and $e' = \approved$. By
inversion of \textsc{T-Adjudicate}, $\tau = \Outcome_D$. By
\textsc{T-Approved}, $\emptyset \vdash_D \approved : \Outcome_D$.
The rule's premises fix
$s'.\mathsf{domain} = \applyMutation_D(s.\mathsf{domain}, a, t)
: \State_D$ by the signature of $\applyMutation_D$
(Definition~\ref{def:adjudicable});
$s'.\mathsf{audit} = s.\mathsf{audit} \,{+\!+}\, [\text{approval
entry}]$ has type $\mathsf{List}(\AuditEntry_D)$ by
Definition~\ref{def:audit-entry}; $s'.\mathsf{queue} =
s.\mathsf{queue}$. Hence $s' : \FrameworkState_D$.

\smallskip
\emph{Case \textsc{E-Adjudicate-Reject}.} $e'$ is $\rejected$,
of type $\Outcome_D$ by \textsc{T-Rejected}. The rule leaves
$s.\mathsf{domain}$ unchanged, appends a rejection audit entry,
and preserves the queue; $s' : \FrameworkState_D$.

\smallskip
\emph{Case \textsc{E-Adjudicate-Escalate}.} $e'$ is $\escalated$,
of type $\Outcome_D$ by \textsc{T-Escalated}. The rule leaves
$s.\mathsf{domain}$ unchanged, appends an escalation audit
entry, and appends a pending proposal of type
$\PendingProposal_D$ (Definition~\ref{def:review-queue}) to the
queue; $s' : \FrameworkState_D$.

\smallskip
\emph{Case \textsc{E-Entry}.} $e = \mathsf{entry}(v)$ and
$e' = v$. By inversion of \textsc{T-Entry},
$\emptyset \vdash_D v : \Outcome_D$ with $\tau = \Outcome_D$, so
$\emptyset \vdash_D e' : \tau$. Framework state unchanged.

\smallskip
\emph{Case \textsc{E-Ctx}.} $e = E[e_1]$ and $e' = E[e_1']$ with
$\config{e_1}{s}{\omega} \longrightarrow
 \config{e_1'}{s'}{\omega'}$
by a sub-derivation. By Lemma~\ref{lem:ctx-decomp}, there exists
$\tau_1$ with $\emptyset \vdash_D e_1 : \tau_1$. By the
inductive hypothesis on the sub-derivation,
$\emptyset \vdash_D e_1' : \tau_1$ and
$s' : \FrameworkState_D$. By Lemma~\ref{lem:ctx-decomp} again,
$\emptyset \vdash_D E[e_1'] : \tau$.

\medskip
\noindent\textbf{Progress.}
Suppose $\config{e}{s}{\omega}$ is well-typed with
$\emptyset \vdash_D e : \tau$ and $e$ is not a value. We show
that either a reduction applies or the configuration is
operationally stuck on an undefined invocation of
$\Coord^{\AR}$, $\Commit^{\AR}$, or an agent in $\AgentPop$.
Proceed by structural induction on~$e$.

\smallskip
\emph{Case $e = x$.} Impossible: $e$ is closed.

\smallskip
\emph{Case $e$ is a value.} Excluded by hypothesis.

\smallskip
\emph{Case $e = \mathsf{let}\; x = e_1 \;\mathsf{in}\; e_2$.}
By inversion of \textsc{T-Let},
$\emptyset \vdash_D e_1 : \tau_1$ for some~$\tau_1$. If $e_1$ is
a value, \textsc{E-Let} applies. Otherwise, the inductive
hypothesis on~$e_1$ either supplies a step (lifted by
\textsc{E-Ctx} with
$E = \mathsf{let}\; x = [\cdot] \;\mathsf{in}\; e_2$) or an
operational stuckness that the outer configuration inherits.

\smallskip
\emph{Case $e = \mathsf{coordinate}$.} If
$\Coord^{\AR}(s, \AgentPop, \omega)$ is defined (including every
agent invocation it performs), \textsc{E-Coordinate} applies.
Otherwise the configuration is operationally stuck on an
undefined invocation of $\Coord^{\AR}$ or of an agent in
$\AgentPop$ (Definitions~\ref{def:agent},
\ref{def:coord-ar}).

\smallskip
\emph{Case $e = \mathsf{adjudicate}(e_1, e_2)$.} By inversion of
\textsc{T-Adjudicate},
$\emptyset \vdash_D e_1 : \Roles_D \times \Action_D$ and
$\emptyset \vdash_D e_2 : \Time$.
\begin{itemize}
  \item If $e_1$ is not a value: the IH on~$e_1$ gives a step
        (lifted by \textsc{E-Ctx} with
        $E = \mathsf{adjudicate}([\cdot], e_2)$) or a stuckness.
  \item If $e_1$ is a value and $e_2$ is not: the IH on~$e_2$
        gives a step (lifted by \textsc{E-Ctx} with
        $E = \mathsf{adjudicate}(e_1, [\cdot])$) or a stuckness.
  \item If both are values: by
        Lemma~\ref{lem:canonical-forms}, $e_1 = (r, a)$ for
        $r \in \Roles_D$ and $a \in \Action_D$, and $e_2 = t$
        for a timestamp~$t$. If
        $\Commit^{\AR}(s, a, t, \omega)$ is defined, its output
        falls in exactly one of the three mutually exclusive
        branches of Definition~\ref{def:commit-ar}, so exactly
        one of \textsc{E-Adjudicate-Approve},
        \textsc{E-Adjudicate-Reject},
        \textsc{E-Adjudicate-Escalate} applies. If undefined,
        the configuration is operationally stuck on
        $\Commit^{\AR}$.
\end{itemize}

\smallskip
\emph{Case $e = \mathsf{entry}(e_1)$.} By inversion of
\textsc{T-Entry}, $\emptyset \vdash_D e_1 : \Outcome_D$. If
$e_1$ is a value, \textsc{E-Entry} applies. Otherwise the IH
on~$e_1$ supplies a step (lifted by \textsc{E-Ctx} with
$E = \mathsf{entry}([\cdot])$) or a stuckness.

\medskip
This exhausts the cases of the structural induction, completing
the proof.
\end{proof}

\begin{remark}[Scope of Type Safety]
\label{rem:type-safety-scope}
Theorem~3 is stated at scope $D : \Domain$ because the Progress
case for $\mathsf{coordinate}$ invokes $\Coord^{\AR}$ and
$\AgentPop$, which are defined only over domains carrying role
structure (Definitions~\ref{def:coord-ar},
\ref{def:agent-population}). Type Safety therefore does not
apply at $\Adjudicable$ scope: a domain lacking role structure
has no $\mathsf{coordinate}$ primitive to type. The
$\Commit^{\AR}$-based results (Theorem~1, Theorem~2,
Corollary~1) do not share this restriction.
\end{remark}

\subsection{The Counselor Queue Extension}
\label{app:sec:counselor-queue}
 
We extend the calculus of \cref{sec:preliminaries} with the
machinery needed to formalise the Counselor Queue informally
discussed in \cref{sec:udt,sec:proofs-counselor-queue}, and extend
the metatheory of \cref{sec:metatheory} to cover it.
 
Two features of the base calculus carry the weight of the
extension. The review queue is already part of the framework state
(Definition~\ref{def:review-queue}) and is populated by
\textsc{E-Adjudicate-Escalate}; the extension activates an
already-present data structure rather than introducing a new one.
And the three-valued invariant result
(Definition~\ref{def:inv-result}) already distinguishes
$\escalateRes$ from $\rejectRes$, so the per-invariant decision
whether a failure admits human override is taken by the domain
author at invariant-declaration time; no new escalation-trigger
mechanism is needed.
 
The extension is conservative: a domain that does not instantiate
the extended signature yields a system identical to the one of
\cref{sec:agentic-redux-instance}, and the results of
\cref{sec:metatheory} apply to it unchanged.
 
\subsubsection{Definitions and Reduction Rules}
\label{sec:counselor-defs}
 
\begin{definition}[Signature \AdjudicablePlus]
\label{def:adjudicable-plus}
A structure~$D$ satisfies \AdjudicablePlus{} iff it satisfies
\Adjudicable{} (Definition~\ref{def:adjudicable}) and, in
addition, provides:
\begin{itemize}
  \item a type $\Counselor_D : \Occ$ of counselor identities;
  \item an authorisation predicate
        $\authorized_D : \Counselor_D \to \{\mathsf{true},
        \mathsf{false}\}$.
\end{itemize}
An element $c : \Counselor_D$ with
$\authorized_D(c) = \mathsf{true}$ is called a \emph{counselor}.
\end{definition}
 
\begin{remark}[Policy authority, not type authority]
\label{rem:counselor-policy}
The counselor may commit any $s_c : \State_D$, regardless of
whether $s_c \models \invariants_D$, and regardless of which
action's escalation triggered the counselor's involvement. The
counselor's state argument is well-typed --- this is not an
untyped escape hatch --- but no kernel-level predicate filters it
on policy grounds. Byzantine behaviour is explicitly out of scope;
the counselor's decisions are policy judgements, not adversarial
inputs.
\end{remark}
 
\paragraph{Extended audit entries.}
The audit entry type is extended with two additional shapes.
 
\begin{definition}[Extended audit entry]
\label{def:audit-entry-plus}
$\AuditEntry_D^{+} \;\triangleq\; \AuditEntry_D \;+\;
\mathsf{CommitEntry}_D \;+\; \mathsf{RejectEntry}_D$, where
\begin{align*}
  \mathsf{CommitEntry}_D \;\triangleq\;
  &\mathsf{Id} \times \Time \times \{\cCommitted\}
   \times \Counselor_D \\
  &\times \State_D \times \mathsf{Id} \\
  &\times \mathsf{List}(\mathsf{InvariantId} \times \InvResult),
\end{align*}
\begin{align*}
  \mathsf{RejectEntry}_D \;\triangleq\;
  &\mathsf{Id} \times \Time \times \{\cRejected\}
   \times \Counselor_D \\
  &\times \mathsf{String} \times \mathsf{Id}.
\end{align*}
The second $\mathsf{Id}$ field of each entry records the
identifier of the originating escalation entry. The
$\mathsf{List}(\mathsf{InvariantId} \times \InvResult)$ field of a
$\mathsf{CommitEntry}$ is the \emph{detection vector}: the kernel
walks $\invariants_D$ on the counselor-committed state~$s_c$ and
records, per invariant, the returned $\InvResult$ value.
\end{definition}
 
The detection vector is informational: the kernel does not use it
to decide whether to commit, and its contents do not affect any
reduction rule below. It exists so an auditor reviewing the log
can mechanically identify counselor commits that did not preserve
invariants, and so the harness has a programmatic signal on which
to base operational recovery
(cf.\ Remark~\ref{rem:counselor-policy}).
 
\paragraph{Escalation--pending-proposal Id sharing.}
For the extension, we fix a convention implicit in the base
calculus's presentation of Step~3 of Definition~\ref{def:commit-ar}:
the $\mathsf{Id}$ of the pending proposal appended to
$s.\mathsf{queue}$ by \textsc{E-Adjudicate-Escalate} is the same
as the $\mathsf{Id}$ of the escalation audit entry appended to
$s.\mathsf{audit}$ by the same rule instance. This makes ``the
escalation entry associated with the current pending proposal''
a well-defined phrase.
 
\paragraph{Extended term language and typing.}
The term language gains one new constructor, $\awaitCounselor$, and
two new value forms, $\cCommitted$ and $\cRejected$:
\[
  e \;::=\; \cdots \;\mid\; \awaitCounselor,
  \qquad
  v \;::=\; \cdots \;\mid\; \cCommitted \;\mid\; \cRejected.
\]
Evaluation contexts are unchanged; $\awaitCounselor$ is atomic.
The types gain one constant:
\[
  \tau \;::=\; \cdots \;\mid\; \CounselorOutcome_D,
\]
with $\CounselorOutcome_D \triangleq \{\cCommitted, \cRejected\}$.
 
\begin{mathpar}
\inferrule*[left=\textsc{T-CCommitted}]
  {\ }
  { \Gamma \vdash_D \cCommitted : \CounselorOutcome_D }
 
\inferrule*[left=\textsc{T-CRejected}]
  {\ }
  { \Gamma \vdash_D \cRejected : \CounselorOutcome_D }
 
\inferrule*[left=\textsc{T-AwaitCounselor}]
  {\ }
  { \Gamma \vdash_D \awaitCounselor : \CounselorOutcome_D }
\end{mathpar}
 
\paragraph{Freeze discipline on base rules.}
Each of the four architecture rules of \cref{sec:opsem-arch}
acquires one additional premise, $s.\mathsf{queue} = [\,]$. Writing
the extended \textsc{E-Coordinate} rule explicitly to illustrate:
 
\begin{mathpar}
\inferrule*[left=\textsc{E-Coordinate$^{+}$}]
  { s.\mathsf{queue} = [\,]
  \\
    \Coord^{\AR}(s, \AgentPop, \omega) = (r, a)
  \\
    \omega \Rightarrow_{k} \omega'
  }
  { \config{\mathsf{coordinate}}{s}{\omega}
      \longrightarrow
    \config{(r, a)}{s}{\omega'}
  }
\end{mathpar}
 
\noindent The premise $s.\mathsf{queue} = [\,]$ is added verbatim
to the three \textsc{E-Adjudicate-*} rules
(\cref{sec:opsem-arch}); we write
\textsc{E-Adjudicate-Approve$^{+}$}, etc., for the resulting rules.
All other premises and conclusions are unchanged. While
$s.\mathsf{queue} \ne [\,]$, the four base architecture rules do
not fire; only the counselor rules below, the administrative
rules \textsc{E-Let} and \textsc{E-Entry}, and the
evaluation-context lift \textsc{E-Ctx}, remain available.
 
\paragraph{Counselor reduction rules.}
The counselor consults the oracle. We write $\omega
\leadsto (\cCommitted, c, s_c)$ to denote ``the oracle's next
consumption supplies a counselor-commit directive with counselor
identity~$c$ and committed state~$s_c$,'' and analogously
$\omega \leadsto (\cRejected, c, m)$.
 
\begin{mathpar}
\inferrule*[left=\textsc{E-Counselor-Commit}]
  { s.\mathsf{queue} = [p]
  \\\\
    \omega \leadsto (\cCommitted, c, s_c)
  \\
    \authorized_D(c) = \mathsf{true}
  \\\\
    s'.\mathsf{domain} \;=\; s_c
  \\
    s'.\mathsf{queue} \;=\; [\,]
  \\\\
    s'.\mathsf{audit} \;=\; s.\mathsf{audit} \,{+\!+}\,
    [\,\text{commit entry}\,]
  \\
    \omega \Rightarrow_{k} \omega'
  }
  { \config{\awaitCounselor}{s}{\omega}
      \;\longrightarrow\;
    \config{\cCommitted}{s'}{\omega'}
  }
\end{mathpar}
 
\begin{mathpar}
\inferrule*[left=\textsc{E-Counselor-Reject}]
  { s.\mathsf{queue} = [p]
  \\\\
    \omega \leadsto (\cRejected, c, m)
  \\
    \authorized_D(c) = \mathsf{true}
  \\\\
    s'.\mathsf{domain} \;=\; s.\mathsf{domain}
  \\
    s'.\mathsf{queue} \;=\; [\,]
  \\\\
    s'.\mathsf{audit} \;=\; s.\mathsf{audit} \,{+\!+}\,
    [\,\text{reject entry}\,]
  \\
    \omega \Rightarrow_{k} \omega'
  }
  { \config{\awaitCounselor}{s}{\omega}
      \;\longrightarrow\;
    \config{\cRejected}{s'}{\omega'}
  }
\end{mathpar}
 
\noindent The ``commit entry'' appended by
\textsc{E-Counselor-Commit} is the $\mathsf{CommitEntry}_D$ value
with counselor~$c$, state~$s_c$, reference field $p.\mathsf{id}$
(the originating escalation entry's Id, by the Id-sharing
convention above), and detection vector
$[(\mathrm{id}(\iota_k), \iota_k(s_c))]_{\iota_k \in
\invariants_D}$. The ``reject entry'' appended by
\textsc{E-Counselor-Reject} is the $\mathsf{RejectEntry}_D$ value
with counselor~$c$, reason~$m$, and reference field
$p.\mathsf{id}$. The rules are mutually exclusive: the oracle's
next consumption either supplies a $\cCommitted$-tagged directive
or a $\cRejected$-tagged one, not both.
 
If $s.\mathsf{queue} = [\,]$, neither rule's first premise is
satisfied and the configuration is stuck on $\awaitCounselor$. If
the oracle supplies an unauthorised counselor identity, neither
rule's third premise is satisfied and the configuration is stuck.
Both forms of stuckness are intended; they are the
$\awaitCounselor$ analogues of the operational stuckness of
Theorem~3 on undefined invocations of $\Coord^{\AR}$,
$\Commit^{\AR}$, or an agent.
 
\paragraph{The extended architecture.}
 
\begin{definition}[\AgReduxPlus{}]
\label{def:ar-plus}
For $D : \Domain$ additionally satisfying \AdjudicablePlus, the
\emph{Agentic Redux architecture with Counselor Queue over~$D$}
is the architecture whose reduction rules are the four base
architecture rules of \cref{sec:opsem-arch} with the freeze
premise $s.\mathsf{queue} = [\,]$ added, together with
\textsc{E-Counselor-Commit} and \textsc{E-Counselor-Reject}. We
write $\AgReduxPlus(D)$ for the resulting system.
\end{definition}
 
\begin{lemma}[Queue cardinality bound]
\label{lem:queue-bound}
For every initial framework state $s_0 : \FrameworkState_D$ with
$s_0.\mathsf{queue} = [\,]$, every client program~$e_0$ well-typed
over $\AgReduxPlus(D)$, and every oracle trace $\omega$: every
reachable configuration $s_j$ in $\mathsf{Exec}(e_0, s_0, \omega)$
satisfies $|s_j.\mathsf{queue}| \leq 1$.
\end{lemma}
\begin{proof}
Induction on~$j$. The base case holds by hypothesis. For the
inductive step, the only rule that extends the queue is
\textsc{E-Adjudicate-Escalate$^{+}$}, which appends one pending
proposal and has premise $s_j.\mathsf{queue} = [\,]$, so
$|s_{j+1}.\mathsf{queue}| = 1$. The two counselor rules have
premise $s_j.\mathsf{queue} = [p]$ and set
$s_{j+1}.\mathsf{queue} = [\,]$. All other rules leave the queue
unchanged, preserving the inductive hypothesis.
\end{proof}
 
\begin{lemma}[Extended audit append discipline]
\label{lem:audit-append-plus}
For every reduction step of $\AgReduxPlus(D)$
$\config{e}{s}{\omega} \longrightarrow \config{e'}{s'}{\omega'}$:
\begin{itemize}
  \item if the step is an instance of
        \textsc{E-Adjudicate-Approve$^{+}$},
        \textsc{E-Adjudicate-Reject$^{+}$},
        \textsc{E-Adjudicate-Escalate$^{+}$},
        \textsc{E-Counselor-Commit}, or
        \textsc{E-Counselor-Reject} (possibly lifted through
        \textsc{E-Ctx}), then
        $s'.\mathsf{audit} = s.\mathsf{audit} \,{+\!+}\,
        [\,\text{entry}\,]$ for exactly one entry, fixed by
        Definition~\ref{def:commit-ar} for the three base rules
        and by \cref{sec:counselor-defs} for the two counselor
        rules;
  \item otherwise, $s'.\mathsf{audit} = s.\mathsf{audit}$.
\end{itemize}
\end{lemma}
\begin{proof}
The three \textsc{E-Adjudicate-*$^{+}$} cases inherit from
Lemma~\ref{lem:audit-append}: the freeze premise restricts when
the rules fire but does not alter their audit-appending
conclusions. The two \textsc{E-Counselor-*} rules have explicit
audit-append premises (\cref{sec:counselor-defs}), each
appending exactly one entry. The remaining rules
(\textsc{E-Let}, \textsc{E-Entry}, \textsc{E-Coordinate$^{+}$})
carry the framework state through their conclusions
syntactically. \textsc{E-Ctx} preserves both cases.
\end{proof}
 
\begin{remark}[Rule classification, extended]
\label{rem:rule-classification-plus}
Each rule of $\AgReduxPlus(D)$ is classified as in
Remark~\ref{rem:rule-classification}, by the shape of its
conclusion's domain clause:
\begin{itemize}
  \item the freeze-premised
        \textsc{E-Adjudicate-Approve$^{+}$} is
        applyMutation-applying;
  \item \textsc{E-Adjudicate-Reject$^{+}$},
        \textsc{E-Adjudicate-Escalate$^{+}$}, and
        \textsc{E-Counselor-Reject} are domain-preserving (each
        fixes $s'.\mathsf{domain} = s.\mathsf{domain}$ in its
        conclusion);
  \item \textsc{E-Let}, \textsc{E-Coordinate$^{+}$}, and
        \textsc{E-Entry} are domain-preserving (their
        conclusions leave $s$ unchanged);
  \item \textsc{E-Counselor-Commit} fits neither class: its
        conclusion fixes $s'.\mathsf{domain} = s_c$ for an
        oracle-supplied state $s_c$ that is in general distinct
        from both $s.\mathsf{domain}$ and any
        $\applyMutation_D$-image of it.
\end{itemize}
Lemma~\ref{lem:domain-localization} therefore extends to
$\AgReduxPlus(D)$ for every step except those instantiating
\textsc{E-Counselor-Commit}.
\end{remark}
 
By Lemma~\ref{lem:queue-bound}, the phrase ``the escalation entry
associated with the current pending proposal'' in the
\textsc{E-Counselor-*} rules refers to a unique audit entry.
 
\subsubsection{Preservation of Base Results}
\label{app:counselor-base-preservation}
 
Theorems~1 and~3, Proposition~1, and Corollary~1 extend to
\AgReduxPlus{} with small additions that exploit the parametric
structure of their proofs. Theorem~2 (Audit Log Integrity) does
not extend by the same pattern and is addressed separately in
\cref{sec:counselor-ext-audit}.
 
All four results below presuppose the well-typed configurations
of \cref{def:wt-config} extended to the new term forms, and the
initial-queue hypothesis $s_0.\mathsf{queue} = [\,]$ needed for
Lemma~\ref{lem:queue-bound}.
 
\paragraph{Theorem 1$^{+}$ (Invariant Preservation under
Safe Counselor Commits).}
\emph{For every $D : \Domain$ additionally satisfying
\AdjudicablePlus, every client program~$e_0$ well-typed over
$\AgReduxPlus(D)$, every initial framework state~$s_0$ of type
$\FrameworkState_D$ with $s_0.\mathsf{domain} \models
\invariants_D$ and $s_0.\mathsf{queue} = [\,]$, and every oracle
trace $\omega \in \Oracle$: if every step of
$\mathsf{Exec}(e_0, s_0, \omega)$ that is an instance of
\textsc{E-Counselor-Commit} (possibly lifted through
\textsc{E-Ctx}) produces a post-state $s_{j+1}$ with
$s_{j+1}.\mathsf{domain} \models \invariants_D$, then every
reachable domain state in $\mathsf{Exec}(e_0, s_0, \omega)$
satisfies $\invariants_D$.}
 
\begin{proof}
The induction of the proof of Theorem~1 carries through; we add
two cases to its inductive step.
 
\emph{Case \textsc{E-Counselor-Commit}.} By hypothesis,
$s_{j+1}.\mathsf{domain} \models \invariants_D$.
 
\emph{Case \textsc{E-Counselor-Reject}.} The rule is
domain-preserving
(Remark~\ref{rem:rule-classification-plus}), so by
Lemma~\ref{lem:domain-localization},
$s_{j+1}.\mathsf{domain} = s_j.\mathsf{domain}$. The inductive
hypothesis gives
$s_{j+1}.\mathsf{domain} \models \invariants_D$. This extends
Case~1 of the proof of Theorem~1 to the new rule.
 
The base cases (\textsc{E-Adjudicate-Approve$^{+}$} and
everything else) are as in the proof of Theorem~1; the added
freeze premise restricts when the base rules fire but does not
alter their conclusions.
\end{proof}
 
\paragraph{Proposition 1$^{+}$ (Agent Confinement, extended).}
\emph{Proposition~1 holds for $\AgReduxPlus(D)$ verbatim.}
 
\begin{proof}
\emph{Part~(i).} The statement and proof depend only on the type
signature of $\alpha_r$ (Definition~\ref{def:agent}), unchanged
by the extension. A counselor is not an agent --- the domain
signature treats $\Counselor_D$ and $\Roles_D$ as disjoint sorts
--- so no counselor-related reduction rule is within the scope
of the proposition.
 
\emph{Part~(ii).} \textsc{E-Coordinate$^{+}$} has the same
conclusion $\config{(r, a)}{s}{\omega'}$ as
\textsc{E-Coordinate}; the added freeze premise restricts when
the rule fires but leaves its conclusion --- and in particular
the syntactic equality $s' = s$ --- unchanged.
\end{proof}
 
\paragraph{Corollary 1$^{+}$ (Write Skew Freedom under Safe Counselor
Commits).}
\emph{Under the hypotheses of Theorem~1$^{+}$,
$\mathsf{Exec}(e_0, s_0, \omega)$ is not a write-skew execution
(Definition~\ref{def:write-skew-calc}).}
 
\begin{proof}
Immediate from Theorem~1$^{+}$, as in the proof of Corollary~1.
\end{proof}
 
\paragraph{Theorem 3$^{+}$ (Type Safety, extended).}
\emph{For every $D : \Domain$ additionally satisfying
\AdjudicablePlus, every well-typed configuration over
$\AgReduxPlus(D)$ satisfies:}
\begin{itemize}
  \item \emph{\textup{(Preservation)} As in Theorem~3, with
        additional cases for the counselor rules.}
  \item \emph{\textup{(Progress)} As in Theorem~3, with two
        additional forms of operational stuckness: an
        $\awaitCounselor$ term with $s.\mathsf{queue} = [\,]$,
        and an $\awaitCounselor$ term for which the oracle's
        next consumption supplies an unauthorised counselor
        identity.}
\end{itemize}
 
\begin{proof}
Preservation gains two cases.
\textsc{E-Counselor-Commit} reduces
$\awaitCounselor : \CounselorOutcome_D$ to $\cCommitted :
\CounselorOutcome_D$; the framework state components preserve
their types by construction of the rule
(Definition~\ref{def:audit-entry-plus}).
\textsc{E-Counselor-Reject} is analogous.
 
Progress gains one case: $e = \awaitCounselor$, of type
$\CounselorOutcome_D$ by \textsc{T-AwaitCounselor}. If
$s.\mathsf{queue} = [p]$ and the oracle's next consumption
supplies an authorised counselor identity, exactly one of the
two counselor rules applies, discriminated by the $\cCommitted$
vs.~$\cRejected$ tag of the oracle directive. Otherwise, the
configuration is operationally stuck.
 
All other cases are as in the proof of Theorem~3. The added
freeze premise on the four base architecture rules does not
introduce a new class of stuckness: a configuration with
$e = \mathsf{coordinate}$ or
$e = \mathsf{adjudicate}((r,a), t)$ and $s.\mathsf{queue} \ne
[\,]$ is stuck in the same sense that the base calculus admits
--- no rule applies --- and this is the client's signal to
interpose an $\awaitCounselor$ term before further
architecture-primitive activity.
\end{proof}
 
\subsubsection{Extended Audit Log Integrity}
\label{sec:counselor-ext-audit}
\label{app:proof-audit-integrity-ext}
 
Theorem~2 (Audit Log Integrity) does not extend to
$\AgReduxPlus(D)$ by the parametric pattern used for the other
base results. Two features of its statement explain why.
 
First, Clauses~(iii) and~(v) of Theorem~2 are rule-list claims:
they name \textsc{E-Adjudicate-Approve}, \textsc{E-Adjudicate-Reject},
and \textsc{E-Adjudicate-Escalate} by name as the rules that
append to the audit and under which per-step exactness holds. The
extension introduces two further audit-appending rules, and the
rule-list claims are structurally incomplete as stated.
 
Second, Theorem~2 addresses a single commit-capable principal: the
meta-agent, whose entries have the shape fixed by
Definition~\ref{def:audit-entry}. The extension introduces a
second commit-capable principal, the counselor, whose entries have
a different shape
(Definition~\ref{def:audit-entry-plus}) and carry provenance
information absent in the base: counselor identity, reference to
the originating escalation entry, per-invariant detection vector
on a post-hoc committed state. Multi-principal auditability
additionally requires a trace-level property that Theorem~2 does
not assert: every counselor entry in the log is the resolution of
an escalation entry earlier in the log. This is the property that
makes the extended log auditable in the operational sense ---
every counselor action can be traced back to the adjudication that
occasioned it.
 
We therefore state the extended auditability result as a new
theorem. The base Theorem~2 remains valid for the subset of
reduction rules it addresses; Theorem~4 covers the full rule set
of $\AgReduxPlus(D)$ and adds the provenance-chaining property.
 
\paragraph{Theorem 4 (Extended Audit Log Integrity).}
\emph{For every $D : \Domain$ additionally satisfying
\AdjudicablePlus, every client program~$e_0$ well-typed over
$\AgReduxPlus(D)$, every initial framework state~$s_0$ of type
$\FrameworkState_D$ with $s_0.\mathsf{queue} = [\,]$, and every
oracle trace $\omega \in \Oracle$, the execution
$\mathsf{Exec}(e_0, s_0, \omega)$ satisfies:}
 
\emph{\textbf{(i)--(ii) Kernel-path faithful entries.} Clauses~(i)
and~(ii) of Theorem~2 hold verbatim, with the
\textsc{E-Adjudicate-*} rules replaced by their freeze-premised
counterparts \textsc{E-Adjudicate-*$^{+}$}.}
 
\emph{\textbf{(iii$^{+}$) Exclusivity of provenance.} For every
reduction step that is not an instance of any of
\textsc{E-Adjudicate-Approve$^{+}$},
\textsc{E-Adjudicate-Reject$^{+}$},
\textsc{E-Adjudicate-Escalate$^{+}$},
\textsc{E-Counselor-Commit}, or \textsc{E-Counselor-Reject}
(possibly lifted through \textsc{E-Ctx}):
$s_{j+1}.\mathsf{audit} = s_j.\mathsf{audit}$.}
 
\emph{\textbf{(iv) Append-only monotonicity.} For every reduction
step in the execution, $s_j.\mathsf{audit}$ is a prefix of
$s_{j+1}.\mathsf{audit}$.}
 
\emph{\textbf{(v$^{+}$) Per-step exactness.} For every reduction
step that is an instance of any of the five rules listed in
Clause~(iii$^{+}$) (possibly lifted through \textsc{E-Ctx}):
$|s_{j+1}.\mathsf{audit}| = |s_j.\mathsf{audit}| + 1$.}
 
\emph{\textbf{(vi) Faithful counselor-commit entries.} For every
step that is an instance of \textsc{E-Counselor-Commit}
(possibly lifted through \textsc{E-Ctx}) with oracle consumption
$(\cCommitted, c, s_c)$ and pending proposal~$p$: the sole entry
appended to $s_{j+1}.\mathsf{audit}$ is a
$\mathsf{CommitEntry}_D$ with tag $\cCommitted$, counselor~$c$,
committed state~$s_c$, escalation-reference $p.\mathsf{id}$, and
detection vector
$[(\mathrm{id}(\iota_k), \iota_k(s_c))]_{\iota_k \in
\invariants_D}$; and $\authorized_D(c) = \mathsf{true}$.}
 
\emph{\textbf{(vii) Faithful counselor-reject entries.} For every
step that is an instance of \textsc{E-Counselor-Reject}
(possibly lifted through \textsc{E-Ctx}) with oracle consumption
$(\cRejected, c, m)$ and pending proposal~$p$: the sole entry
appended to $s_{j+1}.\mathsf{audit}$ is a
$\mathsf{RejectEntry}_D$ with tag $\cRejected$, counselor~$c$,
reason~$m$, and escalation-reference $p.\mathsf{id}$; and
$\authorized_D(c) = \mathsf{true}$.}
 
\emph{\textbf{(viii) Escalation--resolution chaining.} For every
audit entry of tag $\cCommitted$ or $\cRejected$ in $s_j.\mathsf{audit}$
for any~$j$: the entry's escalation-reference field is the
$\mathsf{Id}$ of a unique earlier audit entry in
$s_j.\mathsf{audit}$ of tag $\escalated$, and no other counselor
entry in the entire execution shares that escalation-reference
value. Equivalently: every $\escalated$ entry in the execution
is referenced by at most one counselor entry.}
 
\begin{proof}
Fix a reduction step
$\config{e_j}{s_j}{\omega_j} \longrightarrow
\config{e_{j+1}}{s_{j+1}}{\omega_{j+1}}$
of the execution, derived from a primitive rule possibly lifted
through \textsc{E-Ctx}.
 
\emph{Clauses (i)--(ii).} Inherited from Theorem~2. The proofs
refer to each rule's premises and to
Lemmas~\ref{lem:approval-soundness}
and~\ref{lem:non-approval-witness}; none of these is altered by
the extension, and the added freeze premise does not change the
approval, rejection, or escalation branches' conclusions.
 
\emph{Clauses (iii$^{+}$)--(v$^{+}$).} Immediate from
Lemma~\ref{lem:audit-append-plus}: its audit-appending branch
fixes $s_{j+1}.\mathsf{audit} = s_j.\mathsf{audit} \,{+\!+}\,
[\,\text{entry}\,]$ for exactly one entry, supplying the prefix
property of Clause~(iv) and the cardinality identity of
Clause~(v$^{+}$); the complementary branch supplies
Clause~(iii$^{+}$) and the trivial-prefix case of Clause~(iv).
 
\emph{Clause (vi).} \textsc{E-Counselor-Commit}'s premises fix
$s_{j+1}.\mathsf{audit} = s_j.\mathsf{audit} \,{+\!+}\,
[\,\text{commit entry}\,]$, where the commit entry's shape is
fixed by Definition~\ref{def:audit-entry-plus}: tag $\cCommitted$,
counselor $c$ from the oracle consumption, committed state $s_c$
from the oracle consumption, detection vector computed by the
kernel on $s_c$ (\cref{sec:counselor-defs}), and
escalation-reference $p.\mathsf{id}$ where $p$ is the pending
proposal of the rule's first premise. The authorisation
condition $\authorized_D(c) = \mathsf{true}$ is an explicit
premise of the rule.
 
\emph{Clause (vii).} Analogous, substituting
\textsc{E-Counselor-Reject} and $\mathsf{RejectEntry}_D$.
 
\emph{Clause (viii).} We argue by induction on~$j$.
 
\emph{Base case} ($j = 0$). $s_0.\mathsf{audit} = [\,]$ (no rule
has yet fired); the claim holds vacuously.
 
\emph{Inductive step.} Suppose the claim holds for $s_j$. Consider
the reduction $s_j \to s_{j+1}$ by rule~$R$.
 
If $R$ is not an audit-appending rule,
$s_{j+1}.\mathsf{audit} = s_j.\mathsf{audit}$ by Clause
(iii$^{+}$), and the claim persists.
 
If $R$ is \textsc{E-Adjudicate-Approve$^{+}$} or
\textsc{E-Adjudicate-Reject$^{+}$}, the appended entry is tagged
$\approved$ or $\rejected$. No counselor entry is added, and no
new $\escalated$ entry is added, so the claim persists.
 
If $R$ is \textsc{E-Adjudicate-Escalate$^{+}$}, the appended
entry is a fresh $\escalated$ entry. No counselor entry has been
appended at this step, so the new escalation entry has zero
resolvers in $s_{j+1}.\mathsf{audit}$. All previously existing
$\escalated$ entries retain their resolver counts. All previously
existing counselor entries still reference the same
(unchanged) $\escalated$ entries. The claim persists.
 
If $R$ is \textsc{E-Counselor-Commit} or
\textsc{E-Counselor-Reject}, the appended entry is a counselor
entry with escalation-reference $p.\mathsf{id}$, where $p =
s_j.\mathsf{queue}(0)$. By Lemma~\ref{lem:queue-bound} and the
Id-sharing convention (\cref{sec:counselor-defs}),
$p.\mathsf{id}$ is the Id of the unique $\escalated$ entry in
$s_j.\mathsf{audit}$ whose appending populated the queue with
its current element. That escalated entry has no other
counselor resolver in $s_j.\mathsf{audit}$: any such prior
resolver would have emptied the queue in its step
(Lemma~\ref{lem:queue-bound}), after which the only way the
current~$p$ could be in $s_j.\mathsf{queue}$ is for a subsequent
\textsc{E-Adjudicate-Escalate$^{+}$} step to have appended it,
producing a different escalation entry with a different Id.
Hence the new counselor entry's escalation-reference is the
first such reference in the execution, and the claim persists.
 
This completes the induction.
\end{proof}
 
\begin{remark}[Linear auditability restored]
\label{rem:extended-linearity}
Clauses~(iii$^{+}$)--(v$^{+}$) of Theorem~4 make precise the
``linearly in time'' claim of \cref{sec:linear-auditability} for
the extended system: the audit log grows by exactly one entry per
commit-rule step --- kernel or counselor --- from no other
reduction rule, and existing entries are never modified. Clause
(viii) adds the trace-level provenance that multi-principal
auditability requires: every counselor action in the log points
back to a unique adjudication step earlier in the log, and no
escalation is resolved twice.
\end{remark}
 
\begin{remark}[Detection-vector auditability]
\label{rem:detection-vector-audit}
Clause (vi) establishes that every $\cCommitted$ entry records
the kernel's evaluation of $\invariants_D$ on the
counselor-committed state. A human auditor reviewing the log can
therefore identify, mechanically, every counselor commit that
did not preserve invariants --- an entry is invariant-preserving
iff every $\InvResult$ in its detection vector is $\pass$. This
is the formal hook on which an operational recovery protocol in
the harness can be built (cf.\ Remark~\ref{rem:counselor-policy}),
and it is what makes the conditional hypothesis of
Theorem~1$^{+}$ effective: counselor-compliance violations are
not only real but visibly real in the log.
\end{remark}

\end{document}